\documentclass[11pt]{article}

\usepackage{amssymb,amsmath,amsthm}
\usepackage{graphicx,color}

\newtheorem{theorem}{Theorem}
\newtheorem{observation}{Observation}

\newtheorem{corollary}{Corollary}
\newtheorem{lemma}{Lemma}
\newtheorem{claim}{Claim}

\newcommand{\sm}{\setminus}

\begin{document}

\title{Generating All Minimal Edge Dominating Sets \\
with Incremental-Polynomial Delay\thanks{This work is supported by the European Research Council, the Research Council of Norway, and the French National Research Agency.} }

\author{
Petr A. Golovach\thanks{Department of Informatics, University of Bergen, Norway, 
\texttt{\{petr.golovach, pinar.heggernes,} \texttt{yngve.villanger\}@ii.uib.no}.}
\addtocounter{footnote}{-1}
\and
Pinar Heggernes\footnotemark
\and 
Dieter Kratsch\thanks{LITA, Universit\'e de Lorraine - Metz, France, \texttt{kratsch@univ-metz.fr}.}
\addtocounter{footnote}{-2}
\and
Yngve Villanger\footnotemark\addtocounter{footnote}{-1}
}

\date{}

\maketitle

\begin{abstract}
For an arbitrary undirected simple graph $G$ with $m$ edges, we give an algorithm with running time 
$O(m^4 |\mathcal{L}|^2)$ 
to generate the set $\mathcal{L}$ of all minimal edge dominating sets of $G$. For bipartite graphs we obtain a better result; we show that their minimal edge dominating sets  can be enumerated in time 
$O(m^4 |\mathcal{L}|)$. 
In fact our results are stronger; both algorithms generate the next minimal edge dominating set with incremental-polynomial delay 
$O(m^5 |\mathcal{L}|)$
and 
$O(m^4 |\mathcal{L}|)$
respectively, when $\mathcal{L}$ is the set of already generated minimal edge dominating sets. Our algorithms are tailored for and solve the equivalent problems of enumerating minimal (vertex) dominating sets of line graphs and line graphs of bipartite graphs, with incremental-polynomial delay, and consequently in output-polynomial time. Enumeration of minimal dominating sets in graphs has very recently been shown to be equivalent to enumeration of minimal transversals in hypergraphs. 
The question whether the minimal transversals of a hypergraph can be enumerated in output-polynomial time is a fundamental and challenging question in Output-Sensitive Enumeration; it has been open for several decades and has triggered extensive research in the field.
 
To obtain our results, we present a flipping method to generate all minimal dominating sets of a graph. Its basic idea is to apply a flipping operation to a minimal dominating set $D^*$ to generate minimal dominating sets $D$ such that $G[D]$ contains more edges than $G[D^*]$. Our flipping operation replaces an isolated vertex of $G[D^*]$ with a neighbor outside of $D^*$, and updates $D^*$ accordingly to obtain $D$. The process starts by generating all maximal independent sets, which are known to be minimal dominating sets. Then the flipping operation is applied to every appropriate generated minimal dominating set.  We show that the flipping method for enumeration of minimal dominating sets works successfully for line graphs, resulting in an algorithm with incremental-polynomial delay 
$O(n^2m^2|\mathcal{L}|)$ 
on line graphs and an algorithm with incremental-polynomial delay 
$O(n^2 m|\mathcal{L}|)$
on line graphs of bipartite graphs, where $n$ is the number of vertices and $\mathcal{L}$ is the set of already 
generated minimal dominating sets of the input graph. Finally we show that the flipping method also works for graphs of large girth, resulting in an algorithm with incremental-polynomial delay  
$O(n^2m|\mathcal{L}|^2)$
to enumerate the minimal dominating sets of graphs of girth at least $7$. All given delay times are also the overall running times of the mentioned algorithms, respectively, when $\mathcal{L}$ is the set of all minimal dominating sets of the input graph.
\end{abstract}

\newpage

\section{Introduction} 

Generating all objects that satisfy a specified property, also called {\em enumeration}, plays a central role in Algorithms and Complexity. When enumerating vertex subsets whose number can be exponential in the size of the input graph, one resorts to {\em output-sensitive analysis} for tractability \cite{AvisF96,EiterG95,EiterG02,EiterGM03,JohnsonP88,KhachyanBBEG08,KhachyanBEG08,LawlerL80,Tarjan73},
where the running time is measured in the size of the input plus the number of generated objects. Algorithms that run in {\em output-polynomial time} list all objects in time that is polynomial in the size of the input plus the size of the output.  For various enumeration problems it has been shown that no output-polynomial time algorithm can exist unless P\,=\,NP \cite{KhachyanBBEG08,KhachyanBEG08,LawlerL80}. 
An even better behavior than output-polynomial  is achieved by algorithms with so called {\em incremental-polynomial delay}, which means that the next object in the list of output objects is generated in time that is polynomial in the size of the input plus the number of already generated objects. Incremental-polynomial delay immediately implies output-polynomial time. For practical applications output-polynomial algorithms might be impractical if all the objects are listed only at the end of the algorithm, whereas algorithms with incremental-polynomial delay ensure that within reasonable time one can either obtain a new object or conclude that the list of objects is exhausted. Furthermore, several applications need only to list a certain number of objects and not all of them, in which case incremental-polynomial delay can be very efficient. 

One of the most classical and widely studied enumeration problems is that of listing all minimal transversals of a hypergraph, i.e., minimal hitting sets of its set of hyperedges. This problem has applications in areas like database theory, machine learning, data mining, game theory, artificial intelligence, mathematical programming, and distributed systems; extensive lists of corresponding references are provided by e.g.,~Eiter and Gottlob \cite{EiterGM03}, and Elbassioni, Makino, and Rauf~\cite{ElbassioniMR09}. Whether or not all minimal transversals of a hypergraph can be listed in output-polynomial time, i.e., in time that is polynomial in the size of the hypergraph plus its minimal transversals, has been identified as a fundamental challenge in a long list of seminal papers, e.g.,~\cite{EiterG95,EiterG02,EiterGM03,ElbassioniMR09,FredmanK96,JohnsonP88,Papadimitriou97}, and it remains unresolved despite continuous attempts since the 1980's. 

Recently Kant{\'e}, Limouzy, Mary, and Nourine \cite{KanteLMN12} have proved that enumerating the minimal transversals of a hypergraph is equivalent to enumerating the minimal dominating sets of a graph. In particular, they show that an output-polynomial time algorithm for enumerating minimal dominating sets in graphs implies an output-polynomial time algorithm for enumerating minimal transversals in hypergraphs.  A set of vertices in a graph is a dominating set if every vertex is either in the set or has a neighbor in the set. Such a set is minimal if no proper subset of it is a dominating set. Dominating sets form one of the best studied notions in Algorithms and Complexity; the number of papers on domination in graphs is in the thousands, and several well known surveys and books are dedicated to the topic (see, e.g.,~\cite{HaynesH98}).  

Given the importance of the hypergraph transversal enumeration problem and the failed attempts to resolve whether it can be solved in output-polynomial time, efforts to identify tractable special cases have been highly appreciated \cite{BorosGH98,BorosHIK97,DomingoMP99,Eiter94,EiterG95,EiterGM03,ElbassioniMR09,MakinoI97,MakinoI98}. The newly proved equivalence to domination allows for new ways to attack this long-standing open problem. In fact some results on output-polynomial algorithms to enumerate minimal dominating sets in graphs already exist for graphs of bounded treewidth and of bounded clique-width \cite{Courcelle09}, interval graphs \cite{EiterG95}, strongly chordal graphs \cite{EiterG95},  planar graphs \cite{EiterGM03}, degenerate graphs \cite{EiterGM03}, and split graphs \cite{KanteLMN11}.

In this paper we show that the minimal dominating sets of line graphs and of graphs of large girth can be enumerated with incremental-polynomial delay. More precisely, we give algorithms whose incremental-polynomial delay  is 
$O(n^2m^2|\mathcal{L}|)$
on line graphs,  
$O(n^2 m|\mathcal{L}|)$
on line graphs of bipartite graphs, and  
$O(n^2m|\mathcal{L}|^2)$ 
on graphs of girth at least $7$, where $\mathcal{L}$ is the set of already generated minimal dominating sets of an input graph on $n$ vertices and $m$ edges. 
Line graphs form one of the oldest and most studied graph classes \cite{HararyN60,HemmingerB78,Krausz43,Whitney32} and can be recognized in linear 
time \cite{Roussopoulos73}.
Our results, in addition to proving tractability for two substantial cases of the hypergraph transversal enumeration problem, imply enumeration algorithms with incremental-polynomial delay of minimal edge dominating sets in {\em arbitrary} graphs.  
In particular, we obtain an algorithm with delay 
$O(m^5 |\mathcal{L}|)$
to enumerate the minimal edge dominating sets of any graph on $m$ edges, where $\mathcal{L}$ is the set of already generated edge dominating sets. For bipartite graphs, we are able to reduce the delay to linear in $\mathcal{L}$: 
$O(m^4 |\mathcal{L}|)$.
The given delay is also the asymptotic overall running time for each of the algorithms.

As the central tool in our algorithms, we present a {\em flipping} approach to generate new minimal dominating sets from a parent minimal dominating set.  The basic idea is to first enumerate all maximal independent sets of the input graph $G$ using the algorithm of Johnson, Papadimitriou, and Yannakakis \cite{JohnsonP88}, and then to apply a flipping operation to every appropriate minimal dominating set found, to find new minimal dominating sets inducing subgraphs with more edges. Starting from a parent minimal dominating set $D^*$, our flipping operation replaces an isolated vertex of $G[D^*]$ with a neighbor outside of $D^*$, and supplies the resulting set with necessary additional vertices to obtain new minimal dominating sets $D$ such that $G[D]$ has more edges compared to $G[D^*]$. We show that on all graphs, we can identify a unique parent for each minimal dominating set. On line graphs and graphs of girth at least 7, we are able to prove additional (different) properties of the parents, which allow 
us to obtain the desired running time on these graph classes.

\section{Definitions and Preliminary Results}\label{sec:defs}
We consider finite undirected (if it is not stated explicitly otherwise) graphs without loops or multiple
edges. Given such a graph $G=(V,E)$, its vertex and edge sets, $V$ and $E$, are also denoted by $V(G)$ and $E(G)$, respectively.
The subgraph of $G$ induced by a subset $U\subseteq V$ is denoted by $G[U]$. 
For a vertex $v$, we denote by $N(v)$ its
\emph{(open) neighborhood}, that is, the set of vertices that are
adjacent to $v$. The \emph{closed neighborhood} of $v$ is the set
$N(v)\cup\{v\}$, and it is denoted by $N[v]$. If $N(v)=\emptyset$ then $v$ is {\em isolated}.
For a set $U\subseteq V$, $N[U]=\cup_{v\in U}N[v]$, and $N(U)=N[U] \sm U$. 
The \emph{girth} $g(G)$ of a graph $G$ is the length of a shortest cycle in $G$; if $G$ has no cycles, then $g(G)=+\infty$.  A set of vertices is a {\em clique} if it induces a complete subgraph of $G$. A clique is {\em maximal} if no proper superset of it is a clique.

Two edges in $E$ are adjacent if they share an endpoint.
The \emph{line graph $L(G)$ of $G$} is the graph whose set of vertices is $E(G)$, such that two vertices $e$ and $e'$ of $L(G)$ are adjacent if and only if $e$ and $e'$ are adjacent edges of $G$. A graph $H$ is a \emph{line graph} if $H$ is isomorphic to $L(G)$ for some graph $G$. 
Equivalently, a graph is a line graph if its edges can be partitioned into maximal cliques such that no vertex lies in more than two maximal cliques. This implies in particular that the neighborhood of every vertex can be partitioned into at most two cliques. It is well known that line graphs do not have induced subgraphs isomorphic to $K_{1,3}$, also called a {\it claw}.

Vertex $v$ \emph{dominates} vertex $u$  if $u\in N(v)$; similarly $v$ dominates a set of vertices $U$ if $U\subseteq N[v]$. 
For two sets $D,U\subseteq V$, $D$ dominates $U$ if $U\subseteq N[D]$. A set of vertices $D$ is a \emph{dominating set} of $G=(V,E)$ if $D$ dominates $V$. A dominating set is \emph{minimal} if no proper subset of it is a dominating set. 
Let $D$ be a dominating set of $G$, and let $v \in D$. Vertex $u$ is a \emph{private vertex}, or simply \emph{private}, for vertex $v$ (with respect to $D$) if $u$ is dominated by $v$ but is not nominated by $D \sm \{v\}$. Clearly, $D$ is a minimal dominating set if and only if each vertex of $D$ has a private vertex.  We denote by $P_D[v]$ the set of all private vertices for $v$. Notice that a vertex of $D$ can be private for itself. Vertex $u$ is a \emph{private neighbor} of $v\in D$ if $u\in N(v)\cap P_D[v]$. The set of all private neighbors of $v$ is denoted by $P_D(v)$. Note that $P_D[v]=P_D(v)\cup \{v\}$ if $v$ is isolated in $G[D]$, and otherwise $P_D[v]=P_D(v)$.

A set of edges $A\subseteq E$ is an \emph{edge dominating} set if each edge $e\in E$ is either in $A$ or is adjacent to an edge in $A$. An edge dominating set is {\em minimal} if no proper subset of it is an edge dominating set.
It is easy to see that $A$ is a (minimal) edge dominating set of $G$ if and only if $A$ is a (minimal) dominating set of $L(G)$. 

Let $\phi(X)$ be a property of a set of vertices or edges $X$ of a graph (e.g., ``$X$ is a minimal dominating set''). 
The \emph{enumeration problem for property $\phi(X)$} for a given graph $G$ asks for the set $\mathcal{C}$ of all subsets of vertices or edges $X$ that satisfy $\phi(X)$.  
An \emph{enumeration algorithm} for a set $\mathcal{C}$ is an algorithm that lists the elements of $\mathcal{C}$ without repetitions. 
An enumeration algorithm $\mathcal{A}$ for $\mathcal{C}$ is said to be \emph{output-polynomial} if there is a polynomial $p(x,y)$ such that 
all elements of $\mathcal{C}$ are listed in time $O(p(|G|,|\mathcal{C}|))$. Assume now that $X_1,\ldots,X_\ell$ are the elements of $\mathcal{C}$ 
enumerated in the order in which they are generated by $\mathcal{A}$. 
Algorithm $\mathcal{A}$ enumerates $\mathcal{C}$ with \emph{incremental-polynomial delay} if there is a polynomial $p(x,i)$ such that for each $i\in\{1,\dots,\ell\}$, $X_i$ is generated in time $O(p(|G|,i))$. Finally, 
$\mathcal{A}$ enumerates $\mathcal{C}$ with \emph{polynomial delay} if there is a polynomial $p(x)$ such that for each $i\in\{1,\dots,\ell\}$, the time delay between outputting 
$X_{i-1}$ and $X_{i}$ is $O(p(|G|))$.

A set of vertices $U\subseteq V$ is an \emph{independent set} if no two vertices of $U$ are  adjacent in $G$, and an independent set is \emph{maximal} if no proper superset of it is an independent set. The following observation is folklore.

\begin{observation}\label{obs:mis_mds}
Any maximal independent set of a graph $G$ is a minimal dominating set of $G$. Furthermore, the set of all maximal independent sets of $G$ is exactly the set of all its minimal dominating sets $D$ such that $G[D]$ has no edges.
\end{observation}

Tsukiyama et al.~\cite{TsukiyamaIAS77} showed that maximal independent sets can be enumerated with polynomial delay. Johnson et al.~\cite{JohnsonP88} showed that such an enumeration can be done in lexicographic order. 

\begin{theorem}[\cite{JohnsonP88}]\label{thm:mis} 
All maximal independent sets of a graph with $n$ vertices and $m$ edges can be enumerated in lexicographic order with polynomial delay $O(n(m+n\log |\mathcal{I}|))$, where $\mathcal{I}$ is the set of  already generated maximal independent sets.
\end{theorem}

Let $v_1,\ldots,v_n$ be the vertices of a graph $G$.
Suppose that $D'$ is a dominating set of $G$. We say that a minimal dominating set $D$ is obtained from $D'$ by \emph{greedy removal of vertices} ({\em with respect to order $v_1,\ldots,v_n$}) if 
we initially let $D=D'$ and then recursively apply the following rule:

\medskip 
\noindent
{\em If $D$ is not minimal, then find a vertex $v_i$ with the smallest index $i$ such that
 $D\setminus\{v_i\}$ is a dominating set in $G$, and set $D=D\setminus\{v_i\}$.}

\medskip
\noindent Clearly, when we apply 
this rule,  we never remove vertices of $D'$ that have private neighbors. 
   
\section{Enumeration by flipping: the general approach}\label{sec:main} 
In this section we describe the general scheme of our enumeration algorithms. Let $G$ be a graph; we fix an (arbitrary) order of its vertices: $v_1,\ldots,v_n$. Observe that this order induces a lexicographic order on the set $2^{V(G)}$. Whenever greedy removal of vertices of a dominating set is performed further in the paper,  it is done with respect to this ordering.

\begin{figure}[ht]
\centering\scalebox{0.7}{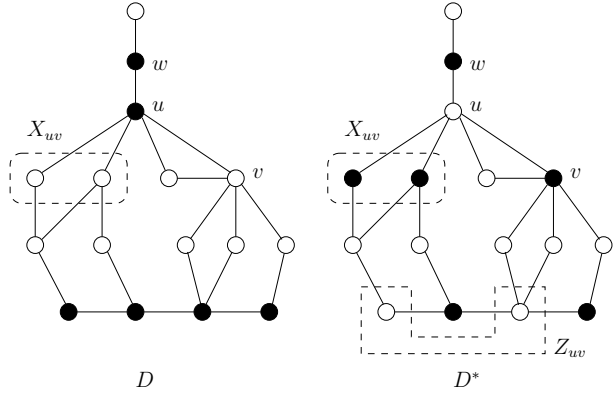}
\caption{A minimal dominating set $D$ and its parent $D^*$; the vertices of $D$ and $D^*$ are black. 
\label{fig:parent}}
\end{figure}

Let $D$ be a minimal dominating set of $G$ such that $G[D]$ has at least one edge $uw$. Then vertex $u\in D$ is dominated by vertex $w\in D$. Let $v\in P_D(u)$.  
Let  $X_{uv}\subseteq P_D(u)\setminus N[v]$ be a maximal independent set in $G[P_D(u)\setminus N[v]]$ selected greedily
 with respect to ordering $v_1,\ldots,v_n$, i.e., we initially set $X_{uv}=\emptyset$ and then recursively include in $X_{uv}$ the vertex of 
$P_D(u)\setminus (N[\{v\}\cup X_{uv}])$ with the smallest index as long as it is possible. 
Consider the set $D'=(D\setminus\{u\})\cup X_{uv}\cup \{v\}$. Notice that $D'$ is a dominating set in $G$, since all vertices of $P_D(u)$ are dominated by $X_{uv}\cup\{v\}$. 
Let $Z_{uv}$ be the set of vertices 
%PG:
that are removed to ensure minimality, 
and let   $D^*=((D\setminus \{u\})\cup X_{uv}\cup \{v\})\setminus Z_{uv}$. See Figure~\ref{fig:parent} for an example.

\begin{lemma}\label{lem:parent}
The set $D^*$ is a minimal dominating set in $G$ such that $X_{uv}\cup\{v\}\subseteq D^*$,
$|E(G[D^*])|<|E(G[D])|$ and $v$ is an isolated vertex of $G[D^*]$.
\end{lemma}

\begin{proof}
Since vertices of $X_{uv}\cup\{v\}$ are privates for themselves, they are not removed by the greedy removal procedure.  Notice that  $E(G[D^*])\subseteq E(G[(D\rq{}])\subseteq E(G[D])\setminus\{uw\}$ as $G[X_{uv} \cup \{v\}]$ is an independent 
set and vertices $X_{uv} \cup \{v\}$ have no neighbors in $D \sm \{u\}$.
Finally $v$ is an isolated vertex as there are no edges incident to $v$ in $G[D\rq{}]$.
\end{proof}

Our main tool, the {\em flipping} operation is exactly the {\em reverse} of how we generated $D^*$ from $D$; i.e., it replaces an isolated vertex $v$ of $G[D^*]$ with a neighbor $u$ in $G$ to obtain $D$.  In particular, we are interested in all minimal dominating sets $D$ that can be generated from $D^*$ in this way. 
%PG:
%Recall that our aim is to start with minimal dominating sets that induce edgeless graphs and generate minimal dominating sets whose induced graphs contain more edges.

Given $D$ and $D^*$ as defined above, we say that $D^*$ is a \emph{parent of $D$ with respect to flipping $u$ and $v$}. We say that $D^*$ is a \emph{parent} of $D$ if there are vertices $u,v\in V(G)$ such that  
$D^*$ is a parent with respect to flipping $u$ and $v$. It is important to note that each minimal dominating set $D$ such that $E(G[D])\neq\emptyset$ has a unique parent with respect to flipping of any vertices 
$u\in D\cap N[D\setminus\{u\}]$ and $v\in P_D(u)$,
as both sets $X_{uv}$ and $Z_{uv}$ are lexicographically first sets selected by a greedy algorithm.
Similarly, we say that $D$ is a \emph{child} of $D^*$ (with respect to flipping $u$ and $v$) if $D^*$ is the parent of $D$ (with respect to flipping $u$ and $v$).

Assume that there is an enumeration algorithm $\mathcal{A}$ that, given a minimal dominating set $D^*$ of a graph $G$ such that $G[D^*]$ has isolated vertices, an isolated vertex $v$ of $G[D^*]$, and a neighbor $u$ of $v$ in $G$, generates with polynomial delay a set of minimal dominating sets $\mathcal{D}$ with the property that $\mathcal{D}$ contains all minimal dominating sets $D$ 
that are children of $D^*$ with respect to flipping $u$ and $v$.
In this case we can enumerate all minimal dominating sets of the graph $G$ with $n$ vertices and $m$ edges as follows. 

\medskip
%PG: New text starts from here.
We define a directed graph $\mathcal{G}$ whose nodes are minimal dominated sets of $G$ and we add a special \emph{root} node $r$
(we use the term \emph{nodes} here to distinguish elements of $V(\mathcal{G})$ from the vertices of $G$).
Recall that by Observation~\ref{obs:mis_mds}, maximal independent sets are minimal dominating sets, i.e., they are vertices of $\mathcal{G}$ . We join the root $r$ with all maximal independent sets by arcs. For each minimal dominating set $D^*\in V(\mathcal{G})$, we join it by an arc with every minimal dominating set $D$ if $\mathcal{A}$ generates $D$ from $D^*$ for some choice of $u$ and $v$.  

We run the depth-first search in $\mathcal{G}$ from $r$. Observe that we should not construct $\mathcal{G}$ to do it, as for each node $W\neq r$ of $\mathcal{G}$ we can use $\mathcal{A}$ to generate all out-neighbors of $W$, and we can use the polynomial delay algorithm by  Johnson et al.~\cite{JohnsonP88} (see Theorem~\ref{thm:mis}) to obtain the out-neighbors of $r$.
Hence, we maintain a list $\mathcal{L}$ of minimal dominating sets of $G$ sorted in lexicographic order that are already visited nodes of $\mathcal{G}$.
Also we keep  a stack $\mathcal{S}$ of records $R_W$ for $W\in V(\mathcal{G})$ that are on the path from $r$ to the current node of $\mathcal{G}$.
These records are used to generate out-neighbors.  
The record $R_r$ contains the last generated maximal independent set and the information that is necessary to proceed with the enumeration of maximal independent sets.
Respectively, the records $R_W$ for $W\neq r$ contain the current choice of $u$ and $v$, the last set $D$ generated by $\mathcal{A}$ for the instance $(W,u,v)$, and 
 the information that is necessary for $\mathcal{A}$ to proceed with the enumeration.
 
\begin{lemma}\label{lem:main}
Suppose that $\mathcal{A}$ generates the elements of $\mathcal{D}$ for a triple $(D^*,u,v)$
with polynomial delay $O(p(n,m))$.
Let also  $\mathcal{L}^*$ be the set of all minimal dominating sets.
Then the described algorithm enumerates all minimal dominating sets in the following running time:
\begin{itemize}
\item incremental-polynomial delay is 
$O((p(n,m)+n^2)m|\mathcal{L}|^2)$ 
and total running time is  
$O((p(n,m)+n^2)m|\mathcal{L}^*|^2)$;
\item if for any $D\in\mathcal{D}$, $|E(G[D])|>|E(G[D^*])|$, then the  
enumeration is done
with incremental-polynomial delay 
$O((p(n,m)+n^2)m^2|\mathcal{L}|)$ and in total running time   
$O((p(n,m)+n^2)m|\mathcal{L}^*|^2)$;
\item if $\mathcal{D}$ contains only children of $D^*$ with respect to flipping of $u$ and $v$, then the enumeration is done with incremental-polynomial delay 
$O((p(n,m)+n^2)m|\mathcal{L}|)$
and in total running time
$O((p(n,m)+n^2)m|\mathcal{L}^*|)$.
\end{itemize}
\end{lemma}

\begin{proof}
Recall that any minimal dominating set $D$ with at least one edge has a parent $D^*$  and $|E(G[D^*])|<|E(G[D])|$. 
Because $\mathcal{A}$ generates $D$ for $D^*$, $(D^*,D)$ is an arc in $\mathcal{G}$. It follows that for any 
minimal dominating set $D\in V(\mathcal{G})$ with at least one edge, there is a maximal independent set $I\in V(\mathcal{G})$ such that $I$ and $D$ are connected by a directed path in $\mathcal{G}$. As $(r,I)$ is an arc in $\mathcal{G}$, $D$ is reachable from $r$.  We conclude that the depth-first search algorithm enumerates all vertices of $\mathcal{G}$. 

It remains to evaluate the running time.

To get a new minimal dominating set, we consider the records in $\mathcal{S}$. 
For each record $R_W$ for $W\neq r$, we have at most $m$ possibilities for $u$ and $v$ to get a new set $D$. As soon as a new set is generated it is added to $\mathcal{L}$ unless it is already in $\mathcal{L}$. Hence, we generate at most $m|\mathcal{L}|$ sets for $W$ in time $(p(n,m)+n^2)m|\mathcal{L}|$, as 
each set is generated with polynomial delay $O(p(n,m))$, and after its generation we immediately test whether or not it is already in $\mathcal{L}$, which takes $O(n\log|\mathcal{L}|) =O(n^2)$ time, because $|\mathcal{L}|\leq 2^n$.
For $R_r$, we generate at most $|\mathcal{L}|$ sets.
Because any isolated vertex of $G$ belongs to every maximal independent set, 
each set is generated with polynomial delay $O(n\rq{}(m+n\rq{}\log |\mathcal{L}|))$, i.e., in time $O(n\rq{}(m+n\rq{}^2))$ by Theorem~\ref{thm:mis}, where $n\rq{}$ is the number of non-isolated vertices. As $n\rq{}\leq 2m$, these sets are generated in time $O(n^2m|\mathcal{L}|)$.
Since $|\mathcal{S}|\leq |\mathcal{L}|$,  in time $O((p(n,m)+n^2)m|\mathcal{L}|^2)$ we either obtain a new minimal dominating set or conclude that the list of minimal dominating sets is exhausted.

To get the bound for the total running time, recall that the depth-first search runs in time that is linear in $|E(\mathcal{G})|$. As for each arc we perform   
$O((p(n,m)+n^2)m)$ operations, the total running time is $O((p(n,m)+n^2)m|\mathcal{L}^*|^2)$.

If for any $D\in\mathcal{D}$, $|E(G[D])|>|E(G[D^*])|$, then the  
incremental-polynomial delay is less. To prove it, we observe that the number of edges in any minimal dominating set is at most $m$. Hence, any directed path starting from $r$ in $\mathcal{G}$ has length at most $m$ and, therefore, $|\mathcal{S}|\leq m+1$. By the same arguments as above, we get that 
 in time $O((p(n,m)+n^2)m^2|\mathcal{L}|)$ we either obtain a new minimal dominating set or conclude that the list of minimal dominating sets is complete.

Assume finally that $\mathcal{D}$ contains only children of $D^*$ with respect to flipping of $u$ and $v$.
Since each minimal dominating set $D$ with $E(G[D])\neq\emptyset$ has a unique parent 
with respect to flipping of any vertices $u\in D\cap N[D\setminus\{u\}]$ and $v\in P_D(u)$, each $D$ has at most $m$ parents. 
Hence, we generate at most $m|\mathcal{L}|$ sets until we obtain a new minimal dominating set or conclude that the list is exhausted. 
As to generate a set and check whether it is already listed we spend time $O(p(n,m)+n^2)$, the delay between two consecutive minimal dominating sets that are output is $O((p(n,m)+n^2)m|\mathcal{L}|)$ and the total running time is $O((p(n,m)+n^2)m|\mathcal{L}^*|)$.
\end{proof}

To be able to apply our approach, we have to show how to construct an algorithm, like algorithm $\mathcal{A}$ above, that produces $\mathcal{D}$ with polynomial delay. We will use the following lemma for this purpose.

\begin{lemma}\label{lem:child-obs}
 Let $D$ be a child of $D^*$ with respect to flipping $u$ and $v$;  $D^*=((D\setminus \{u\})\cup X_{uv}\cup \{v\})\setminus Z_{uv}$. Then for every vertex $z\in Z_{uv}$, the following three statements are true:
\begin{enumerate}

\vspace{-2mm}
\item $z\notin N[X_{uv}\cup\{v\}]$, 

\vspace{-2mm}
\item $z$ is dominated by a vertex of $D^*\setminus (X_{uv}\cup\{v\})$,

\vspace{-2mm}
\item there is a vertex  $x\in N[X_{uv}\cup\{v\}]\setminus N[u]$ adjacent to $z$
such that $x\notin N[D^*\setminus(X_{uv}\cup\{v\})]$.
\end{enumerate}
\vspace{-1mm}
Furthermore, for any  $x\in N[X_{uv}\cup\{v\}]\setminus N[u]$ 
such that $x\notin N[D^*\setminus(X_{uv}\cup\{v\})]$, there is a vertex $z\in Z_{uv}$ such that $x$ and $z$ are adjacent.
\end{lemma}

\begin{proof}
{\it 1.} To show that $z \notin N[X_{uv}\cup \{v\}]$, it is sufficient to observe that $v$ and the vertices of $X_{uv}$ are private neighbors of $u \in D$ and, therefore, no vertex of $X_{uv} \cup \{v\}$ is adjacent to a vertex of $D\setminus\{u\}\supseteq Z_{uv}$.

{\it 2.} Vertices of $Z_{uv}$ are removed from $D'$ by greedy removal in order to obtain a minimal dominating set $D^*$. 
Thus, there is at least one vertex in $D^* = D' \sm Z_{uv}$ that dominates $z$. By {\it 1.} this is not a vertex of $X_{uv}\cup\{v\}$.

%PG:
{\it 3.} By {\it 2.} %every vertex 
$z$ has a neighbor in $D^* \cap D$, i.e., $z$ is not a private for itself. As $D$ is a minimal dominating set there exists a vertex $x \in P_D(z)$.
Vertex $z$ is removed from $D'$ by greedy removal, implying that $P_{D'}(z) = \emptyset$, thus we can conclude that 
$x \in N(X_{uv} \cup \{v\})$ as these are the only vertices added to $D$ in order to obtain $D'$. 

Finally, let  $x\in N[X_{uv}\cup\{v\}]\setminus N[u]$ be a vertex
such that $x\notin N[D^*\setminus(X_{uv}\cup\{v\})]$. Observe that $x$ is not dominated by the set $D^*\setminus (X_{uv}\cup\{v\})\cup\{u\}$, but then it should be dominated by a vertex of $Z_{uv}$ in the dominating set $D=(D^*\setminus (X_{uv}\cup\{v\}))\cup\{u\}\cup Z_{uv}$.
\end{proof}

We use this lemma to construct an algorithm for generating $\mathcal{D}$. The idea is to generate $\mathcal{D}$ by considering all possible candidates for $X_{uv}$ and $Z_{uv}$. 
It is interesting to know whether this can be done efficiently in general. On line graphs and graphs of girth at least 7, we are able to prove additional properties of the parent minimal dominating sets which result in efficient algorithms for generating  $\mathcal{D}$, as will be explained in the sections below.

\section{Enumeration of minimal edge dominating sets}\label{sec:edge}
In this section we show that all minimal edge dominating sets of an {\em arbitrary} graph can be enumerated with incremental-polynomial delay. We achieve this by enumerating the minimal dominating sets in line graphs.

 \subsection{Enumeration of minimal dominating sets of line graphs}\label{sec:line}
%PG:
For line graphs, we construct an enumeration algorithm that, given a minimal dominating set $D^*$ of a graph $G$ such that $G[D^*]$ has isolated vertices, an isolated vertex $v$ of $G[D^*]$, and a neighbor $u$ of $v$ in $G$, generates with polynomial delay a set of minimal dominating sets $\mathcal{D}$ that contains all children of $D^*$ with respect to flipping $u$ and $v$ and
has the property: for any $D\in\mathcal{D}$, $|E(G[D])|>|E(G[D^*])|$.

On line graphs, we can prove some additional properties of a parent. Let $D$ be a minimal dominating set of a graph $G$ such that $G[D]$ has at least one edge 
$uw$,
and assume that $v\in P_D(u)$.
Recall that $D^*$ is defined by choosing a maximal independent set  $X_{uv}\subseteq P_D(u)\setminus N[v]$ in $G[P_D(u)\setminus N[v]]$, then considering the set 
$D'=(D\setminus\{u\})\cup X_{uv}\cup \{v\}$, and letting $D^*=D'\setminus Z_{uv}$ where $Z_{uv}\subseteq D\cap D'$.

\begin{lemma}\label{lem:parent-line}
If $G$ is a line graph, then: 
\begin{itemize}

\vspace{-2mm}
\item $X_{uv}=\emptyset$,

\vspace{-2mm}
\item each vertex of $Z_{uv}$ is adjacent to 
%PG:
%at most 
exactly
one vertex of $P_{D^*}(v)\setminus N[u]$,

\vspace{-2mm}
\item each vertex of $P_{D^*}(v)\setminus N[u]$ is adjacent to exactly one vertex of $Z_{uv}$.
\end{itemize}
\end{lemma}

\begin{proof}
Because $G$ is a line graph, the neighborhood of $u$ can be partitioned into two cliques $K_1$ and $K_2$. 
Vertex $v$ is in $P_D(u)$, and for each $x \in P_D(u)$, $xw\notin E(G)$, since $w\in D$.
Assume that $w\in K_1$. Then $P_D(u) \subseteq K_2 \subseteq N[v]$. 
Hence, $X_{uv}\subseteq P_D(u)\setminus N[v]=\emptyset$.

%PG:
By the definition, each vertex of $Z_{uv}$ is adjacent to 
at least one vertex of $P_{D^*}(v)\setminus N[u]$.
Assume that a vertex $z\in Z_{uv}$ is adjacent to  
at least two vertices $x,y\in P_{D^*}(v)\setminus N[u]$. By the construction of $Z_{uv}$, $z$ is adjacent to a vertex $z'\in D^*\setminus\{v\}$.
Notice that $x$ and $y$ are not adjacent to $z'$, since $x,y \in P_{D^*}(v)$, and recall that $x$ and $y$ are not adjacent to $u$.  Since $G$ is a line graph, $x$ and $y$ are adjacent.
Because $v$ is a private vertex for $u$ with respect to $D$, $v$ and $z'$ are not adjacent,
since  $z' \in D \cap D^*$  by Lemma~\ref{lem:child-obs}. 
Also by Lemma~\ref{lem:child-obs}, we know that $v$ and $z$ are not adjacent.
If $uz,uz\rq{}\notin E(G)$, then we get the left graph in Figure~\ref{fig:forb},
if $uz'\in E(G)$ and $uz\notin E(G)$, then we get the center graph in Figure~\ref{fig:forb}, and finally, 
if the edge $uz$ exists,  we get the right graph in Figure~\ref{fig:forb}.
Beineke has shown that none of these graphs can be an induced subgraph of a line graph~\cite{Beineke70}, and hence we obtain a contradiction.

\begin{figure}[ht]
\centering\scalebox{0.7}{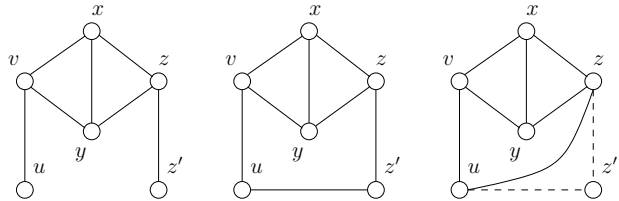}
\caption{Subgraphs $G[\{u,v,x,y,z,z'\}]$ and $G[\{u,v,x,y,z\}]$ that are forbidden induced subgraphs of line graphs~\cite{Beineke70}. 
\label{fig:forb}}
\end{figure}

Finally, observe that each vertex of  $P_{D^*}(v)\setminus N[u]$ is adjacent to at least one vertex of $Z_{uv}$
by Lemma~\ref{lem:child-obs}, because the set of all vertices $x\in N[X_{uv}\cup\{v\}]\setminus N[u]$ 
such that $x\notin N[D^*\setminus(X_{uv}\cup\{v\})]$ is exactly the set  $P_{D^*}(v)\setminus N[u]$. To obtain a contradiction, assume that there is a vertex $x\in P_{D^*}(v)\setminus N[u]$ adjacent to two distinct vertices $y,z\in Z_{uv}$. 
Both vertices $y$ and $z$ belong to $Z_{uv}$ that by definition is a subset of $D$, hence  
$x\notin P_D[y]$ and $x\notin P_D[z]$. 
By the second claim, $x$ is the unique vertex of $P_{D^*}(v)\setminus N[u]$ adjacent to $y$ and $z$.
But since $y\in Z_{uv}$, we get that $P_D(y)\cap(P_{D^*}(v)\setminus N[u])\neq\emptyset$ and $y$ has a neighbor in $P_{D^*}(v)\setminus N[u]$ different from $x$, which gives the desired contradiction.
\end{proof}

%PG: New text.
Consider a line graph $G$ with $n$ vertices $v_1,\ldots,v_n$ and $m$ edges.  
Let $D^*$ be a minimal dominating set and let $v$ be an isolated vertex of $G[D^*]$. Suppose that $u$ is a neighbor of $v$. Let $\{x_1,\ldots,x_k\}=P_{D^*}(v)\setminus N[u]$. 
We construct minimal dominating sets from $(D^*\setminus\{v\})\cup\{u\}$ by adding a set $Z$ that contains a neighbor of each $x_i$ from $N(x_i)\setminus N[v]$. 
Recall that the vertices  $x_1,\ldots,x_k$ should be dominated by $Z_{uv}$ for any child of $D^*$ 
by Lemma~\ref{lem:child-obs}, and by the same lemma each $x_i$ is dominated by a vertex from $N(x_i)\setminus N[v]$. 

Let $U=N[u]\cup (\bigcup_{i=1}^k(N[x_i]\setminus N[v])\cup\{x_i\})$.
We need the following straightforward observation that also will be used in the next section.
To prove it, it is sufficient to notice that because $G$ has no claws, $N[x_i]\setminus N[v]$ is a clique. 

\begin{lemma}\label{lem:U}
For any choice of a set $Z=\{z_1,\ldots,z_k\}$ such that $z_i\in N(x_i)\setminus N[v]$ for $i\in\{1,\ldots,k\}$,  $U$ is dominated by $Z\cup\{u\}$. 
\end{lemma}

We want to ensure that by subsequent removal (to guarantee the minimality) of  vertices of $D^*\setminus\{v\}$, the number of edges in the obtained minimal dominating set is not decreased.
To do it, for each vertex $v_j\in V(G)$, we construct the sets of vertices $R_j$ that 
cannot 
belong to $Z_{uv}$ for any child $D$ of $D^*$, where both $D$ and $D^*$ contain $v_j$.
First, we set   $R_j = \emptyset$ for every $v_j \not\in D^*\setminus\{v\}$.
Let $v_j$ be a vertex of $D^*\setminus\{v\}$ that has a neighbor $v_s$ such that either $v_s\in D^*$ or $v_s=u$. 
As $G$ is a claw-free graph, $K=N(v_j)\setminus N[v_s]$ is a clique. 
Then we set $R_j=K$ in this case.
Notice that we can have several possibilities for $v_s$. In this case $v_s$ is chosen arbitrary.
For all other $v_j\in D^*\setminus\{v\}$, $R_j=\emptyset$. 
Denote by $R$ the set $\cup_{j=1}^n R_j$. 

For each $i\in\{1,\ldots,k\}$, denote by 
$$Z_i=\{z\in V(G) \mid z\in N(D^*\setminus\{v\}) \cap(N(x_i)\setminus (N[v]\cup R)),  N(z)\cap (P_{D^*}(v)\setminus N[u])=\{x_i\}\}.$$ 
We generate a set $\mathcal{D}$ of minimal dominating sets as follows.

\medskip
\noindent
{\bf Case 1.} If at least one of the following three conditions is fulfilled, then we set $\mathcal{D}=\emptyset$:
\begin{itemize}

\vspace{-2mm}
\item[i)] there is a vertex $x\in D^*\setminus\{v\}$ such that $N[x]\subseteq N[ D^*\setminus\{v,x\}]\cup U$,

\vspace{-2mm}
\item[ii)] $k\geq 1$ and there is an index $i\in\{1,\ldots,k\}$ such that $Z_i=\emptyset$,

\vspace{-2mm}
\item[iii)] $u$ is not adjacent to any vertex of $D^*\setminus\{v\}$ and $N(u)\cap (\cup_{i=1}^kZ_i)=\emptyset$.
\end{itemize}

Otherwise, we consider two other cases.

\medskip
\noindent
{\bf Case 2.} If $u$ is adjacent to a vertex of $D^*\setminus\{v\}$, then we consecutively construct all sets $Z=\{z_1,\ldots,z_k\}$ where $z_i\in Z_i$, for $1 \le i \le k$ (if $k=0$, then $Z=\emptyset$). For each $Z$, we construct the set $D'=(D^*\setminus\{v\})\cup\{u\}\cup Z$. Notice that $D'$ is a dominating set as all vertices of $P_{D^*}[v]$ are dominated by $D'$, but $D'$ is not necessarily minimal. Hence, we construct a minimal dominating set $D$ from $D'$ by the greedy removal of vertices.
The obtained set $D$ is unique for a given set $Z$, and it is added to $\mathcal{D}$.

\medskip
Recall that by the definition of the parent-child relation, $u$ should be dominated by a vertex in a child. If $u$ is not adjacent to a vertex of $D^*\setminus\{v\}$, it should be adjacent to at least one of the  added vertices. This gives us the next case.

\medskip
\noindent
{\bf Case 3.} 
If $u$ is not adjacent to any vertex of $D^*\setminus\{v\}$, and $N(u)\cap (\bigcup_{i=1}^k Z_i)\neq\emptyset$, then we proceed as follows. Let $j$ be the smallest index such that $N(u) \cap Z_j \neq \emptyset$, and let $j'$ be the smallest index at least $j$ such that $Z_{j'} \sm N(u) = \emptyset$ ($j'=k$ if they are all non-empty). For each $t$ starting from $t=j$ and continuing until $t=j'$, we do the following. If $N(u) \cap Z_t =\emptyset$ then we go to next step $t=t+1$. Otherwise, for each $w\in N(u) \cap Z_t$, we consider all possible sets $Z=\{z_1,\ldots,z_{t-1},z_{t+1},\ldots, z_k\}\cup\{w\}$ such that $z_i\in Z_i \sm N(u)$ for $1 \le i \le t-1$, and $z_i \in Z_i$ for $t+1 \le i \le k$. As above, for each such set $Z$, we construct the set $D\rq{}=(D^*\setminus\{v\})\cup\{u\}\cup Z$ and then create a minimal dominating set $D$ from $D'$ by the greedy removal of vertices.
The obtained set $D$ is unique for a given set $Z$, and it is added to $\mathcal{D}$.

\medskip
We summarize the properties of the algorithm in the following lemma.

\begin{lemma}\label{lem:gen_from_par}
The set $\mathcal{D}$ is a set of minimal dominating sets such that $\mathcal{D}$ contains all children of $D^*$  with respect to flipping $u$ and $v$, 
for any $D\in \mathcal{D}$, $|E(G[D])|>|E(G[D^*])|$,
and elements of $\mathcal{D}$ are generated with polynomial delay 
$O(n+m)$.
\end{lemma}

\begin{proof}
Notice that each set $D$ constructed in Cases 2 or 3 is a minimal dominating set and that
$\{u\}\cup Z \subseteq D \subseteq D'$. Furthermore every $y\in D'$ such that
$y\in  \{u\}\cup Z$ has a private: $v$ is a private for $u$, and each $x_i$ is a private 
for $z_i$ and $x_t$ is a private for $w$ in Case 3. Observe also that all the constructed sets $D$ are distinct, because they are constructed for distinct sets $Z$. 
We prove the following two claims.

\begin{claim}\label{cl:child}
Each minimal dominating set $D$ that is a child of $D^*$ with respect to flipping $u$ and $v$, is in $\mathcal{D}$.
\end{claim}

\begin{proof}[Proof of Claim~\ref{cl:child}]
%We have to show that each minimal dominating set $D$ that is a child of $D^*$ with respect to flipping $u$ and $v$, is in $\mathcal{D}$.
Let $D$ be a child of $D^*$ with respect to flipping $u$ and $v$. 
Then $D=(D^*\setminus\{v\})\cup\{u\}\cup Z_{uv}$.  
Recall that by Lemma~\ref{lem:parent-line}, $X_{uv}=\emptyset$.
This means that $(D^* \sm \{v\}) \cup \{u\}$ dominates all vertices of $G$ except $P_{D^*}(v) \sm N[u]$. 
As $X_{uv} = \emptyset$, it is clear that each child of $D^*$ with respect to flipping $u$ and $v$, is obtained by adding vertices that dominate $P_{D^*}(v) \sm N[u]$, by Lemma~\ref{lem:child-obs}. 
Also by Lemma~\ref{lem:child-obs}, no vertex of $Z_{uv}$ is a neighbor of $v$, which means that $P_{D^*}(v) \sm N[u]$ 
has to be dominated by  $N(P_{D^*}(v) \sm N[u]) \sm N[v]$. 
Clearly, $(N[v]\setminus N[u])\setminus N[D^*\setminus\{v\})]=
P_{D^*}\setminus N[u]=\{x_1,\ldots,x_k\}$, and for each $z\in Z_{uv}$, there is $x_i$ adjacent to $z$. 
By the last claim of Lemma~\ref{lem:child-obs}, each $x_i$ is adjacent to a vertex of $Z_{uv}$.
Now by Lemma~\ref{lem:parent-line}, each $x_i$ is adjacent to exactly one vertex of $Z_{uv}$, and each vertex $z\in Z_{uv}$ is adjacent to a single vertex of $\{x_1,\ldots,x_k\}$.
Denote by $z_i$ the unique neighbor of $x_i$ in $Z_{uv}$. Clearly, $Z_{uv}=\{z_1,\ldots,z_k\}$. 

By Lemma~\ref{lem:U}, the vertices of the set $U$ are dominated by $Z_{uv}\cup\{u\}$.
Hence, for any vertex $x\in D^*\setminus\{v\}$, $N[x]\setminus (N[ D^*\setminus\{v,x\}]\cup U)\neq\emptyset$ and we do not have Case 1 i).

Now we show that each $z_i\in Z_i$ for $i\in\{1,\ldots,k\}$. 
To obtain a contradiction, assume that some $z_i\notin Z_i$.
Because $z_i\in N(x_i)\setminus N[v]$, $N(z_i)\cap (P_{D^*}(v)\setminus N[u])=\{x_i\}$
and by Lemma~\ref{lem:child-obs}, $z_i$ is dominated by a vertex of $D^*\setminus\{v\}$, we have that $z_i\in R_j\neq\emptyset$ for some $j\in\{1,\ldots,n\}$.
Hence, the vertex $v_j\in D^*$ us adjacent to  some vertex $v_s\in D^*$ or $v_j$ is adjacent to $v_s=u$ and $R_j=N(v_j)\setminus N[v_s]$.
Suppose that $v_j$ is adjacent to $v_s\in D^*$.
Since $v$ is an isolated vertex in $D^*$, $v_s\neq v$.
Then $v_j$ is dominated by $v_s$ and all neighbors of $v_j$ that are not dominated by $v_s$ are in the clique $R_j$, but all these vertices are dominated by $z_i$.
As $v_j$ has no privates in $D$, we have a contradiction. 
Let now $v_j$ be adjacent to $u$. Then we conclude that $N[v_j]\subseteq N[u]\cup N[z_i]$ and again get a contradiction.
We obtain that $Z_i\neq\emptyset$ for $i\in\{1,\ldots,k\}$ and we do not have Case 1 ii).
Recall that for the child $D$ of $D^*$, $u$ is not an isolated vertex of $D\subseteq (D^*\setminus \{v\})\cup(\cup_{i=1}^kZ_i)$. Hence, we do not have Case 1 iii) as well. 
Finally, as $z_i\in Z_i$, $Z_{uv}$ is listed in Case 2 or 3. 
We conclude that each child of $D^*$ with respect to flipping $u$ and $v$, is contained in $\mathcal{D}$. 
\end{proof}

\begin{claim}\label{cl:mon}
 For any $D\in \mathcal{D}$, $|E(G[D])|>|E(G[D^*])|$.
\end{claim}

\begin{proof}[Proof of Claim~\ref{cl:mon}]
%Our next aim is to show that for any $D\in \mathcal{D}$, $|E(G[D])|>|E(G[D^*])|$.
Assume that a minimal dominating set $D\in\mathcal{D}$ is obtained by the greedy removal procedure from
$D\rq{}=(D^*\setminus\{v\})\cup\{u\}\cup Z$ where $Z=\{z_1,\ldots,z_k\}$. 

We show that only isolated vertices of $D^*\setminus\{v\}$ that are not adjacent to $u$ can be removed. To obtain a contradiction, 
assume that some vertex  $v_j\in D^*\setminus\{v\}$ such that $v_j$ has a neighbor $v_s$ where $v_s \in D^*\setminus\{v\}$ or $v_s=u$, and $R_j=N(v_j)\setminus N[v_s]$ is removed.
As $v_j$ is removed, $v_j$ has no privates in $D$. Notice that by the construction of the sets $Z_1,\ldots,Z_k$, $R_j\cap Z=\emptyset$. Observe also that $u\notin R_j$ as we would have
$N[v_j]\subseteq N[v_s]\cup N[u]$, and we would have Case 1 i).
Let $P=N(v_j)\setminus N[D^*\setminus \{v,v_j\}]$. Since $v_j$ has no privates in $D$, each vertex $y\in P$ is dominated by a vertex of $Z$ or by $u$. 
Clearly, if $y$ is dominated by $u$, then $y\in N[u]\subseteq U$. 
Suppose that  $y$ is dominated by $z_i\in Z$ but not $u$. Since $x_i$ is not dominated by $D^*\setminus \{v\}$, $y\neq x_i$. 
We consider two cases.

\medskip
\noindent{Case a).} The vertex $v_j$ has a neighbor $v_s\in D^*\setminus\{v\}$ and $R_j=N(v_j)\setminus N[v_s]$.  
By the construction of $Z_i$, $z_i$ is dominated by some vertex $z\in D^*\setminus\{v\}$.
If $z=v_j$, then because $z_i\notin R_j$, $z_i$ is adjacent to $v_s$. Hence, without loss of generality we can assume that $z\neq v_j$, as otherwise we can take $z=v_s$.
Since $y\in P$, $z$ is not adjacent to $y$, and because $x_i$ is not dominated by $D^*\setminus\{v\}$, $z$ is not adjacent to $x_i$.
The graphs $G$ has no claws. Therefore, $yx_i\in E(G)$. If $y$ is adjacent to $v$, then $G[\{z,z_i,y,x_i,v,u\}]$ or $G[\{z_i,y,x_i,v,u\}]$ is isomorphic to one of the graphs shown in Fig.~\ref{fig:forb} and forbidden for line graphs~\cite{Beineke70}.  Then $y$ is not adjacent to $v$. It means  that $y\in N(x_i)\setminus N[v]\subseteq U$.
Hence, we conclude that $y\in U$ in the considered case.

\medskip
\noindent{Case b).} Some vertex $v_j$ is adjacent to $u$ and $R_j=N(v_j)\setminus N[u]$. 
Suppose that $z_i$ is adjacent to $u$. Recall that $u$ is not adjacent to $y$ and $x_i$. Hence, as $G$ has no claws, $yx_i\in E(G)$.
If $y$ is adjacent to $v$, then $G[\{z_i,y,x_i,v,u\}]$ is forbidden for line graphs~\cite{Beineke70}.  Then $y$ is not adjacent to $v$. It means  that $y\in N(x_i)\setminus N[v]\subseteq U$.
Suppose now that $z_i$ is not adjacent to $u$. 
By the construction of $Z_i$, $z_i$ is dominated by some vertex $z\in D^*\setminus\{v\}$.
If $z=v_j$, then because $z_i\notin R_j$, $z_i$ is adjacent to $u$ and we obtain a contradiction.
Then $z\neq v_j$.
Since $y\in P$, $z$ is not adjacent to $y$, and because $x_i$ is not dominated by $D^*\setminus\{v\}$, $z$ is not adjacent to $x_i$.
The graphs $G$ has no claws. Therefore, $yx_i\in E(G)$. If $y$ is adjacent to $v$, then $G[\{z,z_i,y,x_i,v,u\}]$ or $G[\{z_i,y,x_i,v,u\}]$ is  forbidden for line graphs~\cite{Beineke70}.  Then $y$ is not adjacent to $v$. It means  that $y\in N(x_i)\setminus N[v]\subseteq U$.
Hence, we conclude that $y\in U$ in this case too.

\medskip
\noindent
We have that $P\subseteq U$, but in this case 
$N[v_j]\subseteq N[ D^*\setminus\{v,v_j\}]\cup U$, and we have Case 1 i) of our algorithm.

We proved that only isolated vertices of $D^*\setminus\{v\}$ that are not adjacent to $u$ may be removed by the greedy removal. If $u$ is adjacent to a vertex of $D^*$, then  
 $|E(G[D])|>|E(G[D^*])|$ because we do not destroy edges between the vertices of $D^*\setminus\{v\}$ and add at least one edge incident with $u$ to the constructed set. Suppose that $u$ is not adjacent to the vertices of $D^*\setminus\{v\}$. Then we have Case 3, and $Z$ contains at least one vertex adjacent to $u$, i.e., we increase the number of edges. This observation concludes the proof that $|E(G[D])|>|E(G[D^*])|$.
\end{proof}

To complete the proof, it remains to evaluate the running time. Observe that the sets $Z_1,\ldots,Z_k$ can be constructed in time $O(n+m)$.
The sets $Z$ can clearly be generated with delay $O(n)$, and greedy removal of vertices on each generated set can be done in time $O(n+m)$. Note that each set is generated exactly once.
\end{proof}

Combining Lemmas~\ref{lem:main} and \ref{lem:gen_from_par}, we obtain the following theorem.

\begin{theorem}\label{thm:line}
All minimal dominating sets of a line graph with $n$ vertices and $m$ edges can be enumerated
with incremental-polynomial delay 
$O(n^2m^2|\mathcal{L}|)$,
and in total time 
$O(n^2m|\mathcal{L}^*|)$,
where $\mathcal{L}$ is the set of already generated minimal dominating sets and
$\mathcal{L}^*$ is the set of all minimal dominating sets.
\end{theorem}

This theorem immediately gives us the following corollary.

\begin{corollary}\label{cor:eds}
All minimal edge dominating sets of an arbitrary graph with $m$ edges can be enumerated
with incremental-polynomial delay 
$O(m^5|\mathcal{L}|)$,
and in total time 
$O(m^4|\mathcal{L}^*|^2)$, 
where $\mathcal{L}$ is the set of already generated minimal edge dominating sets and
$\mathcal{L}^*$ is the set of all minimal edge dominating sets.
\end{corollary}

\subsection{Enumeration of minimal dominating sets of line graphs of bipartite graphs}
\label{sec:linebip}
We can improve the dependence on the size of the output if we restrict our attention to edge dominating sets of 
bipartite graphs.  To do it, we construct an algorithm that enumerates children of minimal dominating sets with polynomial delay. 
Again, we work on the equivalent problem of generating minimal dominating sets of line graphs of bipartite graphs. The following observation is not difficult to verify.

\begin{lemma}\label{lem:cliques}
Let $H$ be a non-empty bipartite graph and let $G=L(H)$. Then 
\begin{itemize}

\vspace{-2mm}
\item for each vertex $v$ in $G$, either $N(v)$ is a clique or $N(v)$ is a union of two disjoint cliques $K_1$ and $K_2$ such that no vertex of $K_1$ is adjacent to a vertex of $K_2$, and

\vspace{-2mm}
\item $G$ has no induced cycle on $2k+1$ vertices for any $k>1$.
\end{itemize}
\vspace{-2mm}
Furthermore, if $u$ and $v$ are adjacent vertices of $G$, then for any $x\in N(v)\setminus N[u]$, $N(x)\setminus N[v]$ is a clique, and for any distinct  $x,y\in N(v)\setminus N[u]$, $N(x)\setminus N[v]$ and $N(y)\setminus N[v]$ are disjoint.
\end{lemma}

Let $G$ be the line graph of a bipartite graph $H$. 
Recall that $n$ is the number of vertices and $m$ is the number of edges of $G$.
Let $V(G)=\{v_1,\ldots,v_n\}$.
Let $D^*$ be a minimal dominating set of $G$ that has an isolated vertex $v$ in $G[D^*]$. 
Suppose that $u$ is a neighbor of $v$ in $G$. Let $\{x_1,\ldots,x_k\}=P_{D^*}(v)\setminus N[u]$.

As with line graphs, we construct minimal dominating sets from $(D^*\setminus\{v\})\cup\{u\}$ by adding a
%PG:
set of vertices $Z$ that contains neighbor of each $x_i$. 
Now we have to ensure that neighbors are selected in such a way that the obtained sets are minimal dominating sets that are children of $D^*$.
To do it, for each vertex $v_j\in V(G)$, we construct the sets of vertices that 
cannot 
belong to $Z_{uv}$ for any child $D$ of $D^*$, where both $D$ and $D^*$ contain $v_j$.
For each $v_j$, we define sets $R_j$ and $S_j$, where
the sets $R_j$ are used to ensure minimality, and the sets $S_j$ are used to guarantee that the obtained set is a child of $D^*$. 
First, we set   $R_j = S_j = \emptyset$ for every $v_j \not\in D^*\setminus\{v\}$.
Recall that by Lemma~\ref{lem:U}, the set 
$U=N[u]\cup (\bigcup_{i=1}^k((N[x_i]\setminus N[v])\cup\{x_i\})$ is dominated for every choice of $Z\subseteq V(G)\setminus N[v]$ that dominates $\{x_1,\ldots,x_k\}$, and that  $\{x_1,\ldots,x_k\}$ should be dominated by $Z_{uv}$ for any child of $D^*$,  by Lemma~\ref{lem:child-obs}.

Let $v_j$ be a vertex of $D^*\setminus\{v\}$. 
By Lemma~\ref{lem:cliques}, $N(v_j)$ is a union of at most two disjoint cliques with no edges between them.
We define sets $R_j$ as follows: 
\begin{itemize}

\vspace{-2mm}
\item[i)] if $P_{D^*}[v_j]\setminus U=\{v_j\}$ and
$N(v_j)\cap (N(v)\setminus N[u])=\emptyset$,  then $R_j=N(v_j)$;

\vspace{-2mm}
\item[ii)] if $P_{D^*}[v_j]\setminus U\neq\{v_j\}$, and  
\begin {itemize}

\vspace{-2mm}
\item for any vertex\footnote{Note that $N(v_j)\cap (N(v)\setminus N[u])$ might be empty.} $x\in N(v_j)\cap (N(v)\setminus N[u])$, $x$ is dominated by at least two vertices of $D^*\setminus\{v\}$, and 

\vspace{-2mm}
\item $P_{D^*}(v_j)\setminus U\subseteq K$, where $K$ is a maximal clique in $N(v_j)$, 
\end{itemize}
\vspace{-2mm}
then $R_j=K$;

\vspace{-2mm}
\item[iii)] in all other cases $R_j=\emptyset$.
\end{itemize}

Suppose that $v_j$ is a vertex of $D^*\setminus\{v\}$, such that $v_j$ is adjacent to a vertex $x\in N(v)\setminus N[u]$ that has no other neighbors in $D^*\setminus\{v\}$, and $(P_{D^*}[v_j]\cap N(x))\setminus U\subseteq \{v_j\}$.  For such a vertex $v_j$, we define $S_j=\{v_\ell\in V(G) \mid v_jv_{\ell} \in E(G),\ell>j\}$. For all other vertices $v_j$, we set $S_j=\emptyset$. 

Now for each $i\in\{1,\ldots,k\}$, we construct the set $Z_i$. A vertex $v_s\in N(x_i)\setminus N[v]$ is included in $Z_i$ 
if and only if
$v_s$ is adjacent to a vertex of $D^*\setminus\{v\}$, and
$v_s\notin R_j$ and $v_s\notin S_j$ for all $j\in\{1,\ldots,n\}$.  
Observe that by Lemma~\ref{lem:cliques}, these sets $Z_i$ are disjoint cliques,
and notice that $Z_i \cap D^* = \emptyset$ as $x_i$ 
is in $P_{D^*}(v)\setminus N[u]$. We generate a set $\mathcal{D}$ of minimal dominating sets as follows.

\label{alg:linebip}

\medskip
\noindent
{\bf Case 1.} If at least one of the following three conditions is fulfilled, then we set $\mathcal{D}=\emptyset$:
\begin{itemize}

\vspace{-2mm}
\item[i)] there is a vertex $x\in D^*\setminus\{v\}$ such that $N[x]\subseteq N[D^*\setminus\{v,x\}]\cup U$,

\vspace{-2mm}
\item[ii)] $k\geq 1$ and there is an index $i\in\{1,\ldots,k\}$ such that $Z_i=\emptyset$,

\vspace{-2mm}
\item[iii)] $u$ is not adjacent to any vertex of $D^*\setminus\{v\}$ and $N(u)\cap (\cup_{j=1}^kZ_j)=\emptyset$.
\end{itemize}

Otherwise, we consider two other cases.

\medskip
\noindent
{\bf Case 2.} If $u$ is adjacent to a vertex of $D^*\setminus\{v\}$,
then one after another we consider all possible sets $Z=\{z_1,\ldots,z_k\}$ such that $z_i\in Z_i$ for $1 \le i \le k$ (if $k=0$ then $Z=\emptyset$). For each $Z$, we construct the set $D=(D^*\setminus\{v\})\cup\{u\}\cup Z$ and add it to $\mathcal{D}$.

\medskip

Recall that by the definition of the parent-child relation, $u$ should be dominated by a vertex of $D$. If $u$ is not adjacent to a vertex of $D^*\setminus\{v\}$, it should be adjacent to at least one of the  added vertices. This gives us the next case.

\medskip
\noindent
{\bf Case 3.} 
If $u$ is not adjacent to any vertex of $D^*\setminus\{v\}$, and $N(u)\cap (\bigcup_{i=1}^k Z_i)\neq\emptyset$, then we proceed as follows. Let $j$ be the smallest index such that $N(u) \cap Z_j \neq \emptyset$, and let $j'$ be the smallest index at least $j$ such that $Z_{j'} \sm N(u) = \emptyset$ ($j'=k$ if they are all non-empty). For each $t$ starting from $t=j$ and continuing until $t=j'$, we do the following. If $N(u) \cap Z_t =\emptyset$ then we go to next step $t=t+1$. Otherwise, for each $w\in N(u) \cap Z_t$, we consider all possible sets $Z=\{z_1,\ldots,z_{t-1},z_{t+1},\ldots, z_k\}\cup\{w\}$ such that $z_i\in Z_i \sm N(u)$ for $1 \le i \le t-1$, and $z_i \in Z_i$ for $t+1 \le i \le k$. As above, for each such set $Z$, we construct the set $D=(D^*\setminus\{v\})\cup\{u\}\cup Z$ and add it to $\mathcal{D}$.

\medskip
The correctness of the described algorithm is proved in the following lemma, whose full proof is given in the appendix.

\begin{lemma}\label{lem:bipartite}
The set $\mathcal{D}$ is the set of all children of $D^*$ with respect to flipping $u$ and $v$.
\end{lemma}

\noindent
{\it Proof idea.} The result is immediately implied by the following three claims, which we prove in the appendix. (1) Each $D\in\mathcal{D}$ is a minimal dominating set. (2) Each $D\in\mathcal{D}$ is a child of $D^*$ with respect to flipping $u$ and $v$. (3) If $D$ is a child of $D^*$ with respect to flipping $u$ and $v$, then  $D\in\mathcal{D}$. \qed

\begin{lemma}\label{lem:bipartite-time} The elements of $\mathcal{D}$ can be generated with polynomial delay $O(n+m)$.
\end{lemma}

\begin{proof}
The sets $Z_1,\ldots,Z_k$ can be constructed in $O(n+m)$ time. Within the same time we can also compute $Z_i \cap N(u)$ and $Z_i \sm N(u)$ for $1 \le i \le k$, which gives delay $O(n+m)$ before the first minimal dominating set is generated. In Case 2, we can trivially generate the next set $Z$ in $O(n)$ time. In Case 3, every considered vertex $w$ results in the generation of a new set $Z$. Hence, in both Cases 2 and 3, we can construct the next set $Z$ with delay $O(n)$. Note that by the way we generate sets $Z$, each minimal dominating set is generated exactly once.
\end{proof}

By Lemmas~\ref{lem:main}, \ref{lem:bipartite} and \ref{lem:bipartite-time}, we have the following theorem and corollary.

\begin{theorem}\label{thm:line-bipartite}
All minimal dominating sets of the line graph $G$ of a bipartite graph can be enumerated with incremental-polynomial delay 
$O(n^2m|\mathcal{L}|)$,
and in total time 
$O(n^2m|\mathcal{L}^*|)$,
where $n=|V(G)|$, $m=|E(G)|$, $\mathcal{L}$ is the set of already generated minimal dominating sets, 
and $\mathcal{L}^*$ is the set of all minimal dominating sets.
\end{theorem}

\begin{corollary}\label{cor:eds-bipartite}
All minimal edge dominating sets of a bipartite graph with $m$ edges can be enumerated
with incremental-polynomial delay 
$O(m^4|\mathcal{L}|)$,
and in total time 
$O(m^4|\mathcal{L}^*|)$,
where $\mathcal{L}$ is the set of already generated minimal dominating sets,
and $\mathcal{L}^*$ is the set of all minimal edge dominating sets.
\end{corollary}

\section{Graphs of large girth and concluding remarks}
On line graphs we were able to observe properties of the parent relation in addition to uniqueness, which made it possible to apply the flipping method and design efficient algorithms for enumerating the minimal dominating sets. In Section \ref{sec:girth}, given in the appendix, we present an additional application of the flipping method, and we show that it also works successfully on graphs of girth at least 7. To do this, we observe other desirable properties of the parent relation on this graph class. As a result, we obtain an algorithm that enumerates the minimal dominating sets of a graph of girth at least 7 with incremental-polynomial delay $O(n^2m|\mathcal{L}|^2)$.

\medskip

To conclude, the flipping method that we have described in this paper has the property that each generated minimal dominating set has a unique parent.
It would be very interesting to know whether this can be used to obtain output-polynomial time algorithms for enumerating
minimal dominating sets in general. For the algorithms that we have given in this paper, on the studied graph classes we were
able to give additional properties of the parents to obtain the desired running times. Are there additional properties of parents
in general graphs that can result in efficient algorithms? 

As a first step towards resolving these questions, on which other graph classes can the flipping method be used to enumerate the minimal dominating sets in output-polynomial time? Another interesting question is whether the minimal dominating sets of line graphs or graphs of large girth can be enumerated with polynomial delay.

Recently and independently Kant{\'e} et al.~\cite{KanteLMN12b} have showed that {\it line graphs} and {\it path graphs}
have bounded neighbourhood Helly, and thus minimal dominating sets can be enumerated in output polynomial time.

\newpage

\centerline{\huge {\bf Appendix}}

\medskip

\section*{Proof of Lemma \ref{lem:bipartite}}

The graph obtained from the complete graph on four vertices by the deletion of one edge is called a \emph{diamond}. 
A graph that has no induced subgraph isomorphic to a diamond is said to be \emph{diamond-free}.
It is straightforward to see that the first claim of Lemma~\ref{lem:cliques} implies the following observation. 

\begin{observation}\label{obs:diamond}
Line graphs of bipartite graphs are diamond-free. 
\end{observation}

Let $G$ be the line graph of a bipartite graph $H$, such that $V(G)=\{v_1,\ldots,v_n\}$.
Let $D^*$ be a minimal dominating set of $G$ that contains a vertex $v$ which is isolated in $G[D^*]$,  and let $u$ be a neighbor of $v$ in $G$. 
Let $\mathcal{D}$ be the collection of sets generated from $(D^*\setminus\{v\})\cup\{u\}$ by the algorithm described on page \pageref{alg:linebip} in Section \ref{sec:linebip}.

\bigskip

\noindent
{\bf Lemma \ref{lem:bipartite}.}
{\it The set $\mathcal{D}$ is the set of all children of $D^*$ with respect to flipping $u$ and $v$.}

\smallskip

\begin{proof}
We prove the following three claims.

\begin{claim}\label{cl:A}
Each $D\in\mathcal{D}$ is a minimal dominating set.
\end{claim}

\begin{proof}[Proof of Claim~\ref{cl:A}]
Let $D\rq{}=(D^*\setminus\{v\})\cup\{u\}$. 
Recall that by Case 1 i), if there is a vertex $x\in D^*\setminus\{v\}$ such that $N[x]\subseteq N[ (D^*\setminus\{v,x\})\cup\{u\}]\cup U$, then  $\mathcal{D}=\emptyset$. 
Hence, for each $x\in D\rq{}$ such that $x\neq u$, at least one vertex of $N[x]$ is not dominated by $D\rq{} \sm \{x\}$ and is not included in $U$.
Therefore, it is possible to extend $D\rq{}$ to a minimal dominating set, and we do it by including $Z$.
Let $D=D\rq{}\cup Z$. $D$ is a dominating set of $G$, because $P_{D^*}[v]$ is dominated by $Z\cup\{u\}$. 
To obtain a contradiction, assume that $D$ is not minimal. Then there is a vertex $v_j\in D$ such that $v_j$ has no privates. 
Vertex $v$ is a private for $u$, each vertex $x_i$ is a private for $z_i$, and $x_t$ is a private for $w$ in Case 3. Hence, $v_j\in D^*\setminus\{v\}$.
As $U$ is dominated by any minimal dominating set $D$ generated by our algorithm from $D^*$, we can conclude that at least one vertex of $N[v_j]$ is not dominated by $D\rq{} \sm \{v_j\}$ and is not included in $U$, and 
 $P=P_{D^*}[v_j]\setminus U\neq\emptyset$. 
We have that $P\subseteq N[Z]$ as $P_D[v_j] = \emptyset$.
There is a vertex $z_i\in Z$ such that $z_i$ is adjacent to a vertex $y\in P$, and let us assume for now that $z_iv_j\notin E(G)$.
By the construction of the sets $Z_i$, $z_i$ is adjacent to a vertex $z \in D^*\setminus\{v\}$.
Because $z_iv_j\notin E(G)$, $z\neq v_j$.
Since $x_i$ is dominated only by $z_i$, 
$x_i$ not dominated by $z$. Then either $zy\in E(G)$ or $yx_i\in E(G)$, since $G$ is a line graph. 
Since $y \notin U$ by definition, $yx_i\notin E(G)$ as 
%PG:
$(N[x_i]\setminus N[v])\cup\{x_i\} \subseteq U$.
Thus, we can conclude that $zy\in E(G)$, but then $y$ is dominated by at least two vertices of $D^*\setminus\{v\}$, contradicting that $y \in P_D^*(v_j)$. 
As this argument holds for any vertex $y \in P$, we get that $P$ is dominated by  $Z \cap N(v_j)$.

If $P=\{v_j\}$ and $N(v_j)\cap (N(v)\setminus N[u])=\emptyset$, then $R_j=N(v_j)$,  but this contradicts the definition of $Z_i$, as $Z_i \cap R_i = \emptyset$.
Hence if $P=\{v_j\}$ then there is a vertex $x\in N(v_j)\cap (N(v)\setminus N[u])$. In this case,
because $x$ is not a private for $v_j$ with respect to $D$, $x$ is dominated by a vertex $y\in D$. 
As $x$ is adjacent to both $v,v_j \in D^*$, it is clear that $x \not\in \{x_1,x_2,\ldots, x_k\}$, and by Lemma~\ref{sec:linebip} 
we can conclude that $\{x,x_1,x_2,\ldots, x_k\}$ is a clique and that $N(x) \cap N(x_i) = \emptyset$ for $i \in [1, \ldots, k]$ and thus $y \not\in Z_i$.
Hence, $y\in D^*\setminus\{v,v_j\}$. By Lemma~\ref{lem:cliques}, $yv_j\in E(G)$ and $v_j\notin P$, contradicting that $P=\{v_j\}$.
It follows that  $P$ has at least one vertex in $N(v_j)$.  
By Lemma~\ref{lem:cliques}, $N(v_j)$ is a union of at most two cliques. 
Suppose that $N(v_j)$ is a union of two non-empty cliques $K_1,K_2$ such that there are no edges between $K_1$ and $K_2$, and suppose that $P$ intersects both $K_1$ and $K_2$.
Then there are vertices $z_i,z_{i\rq{}}\in Z$ such that $z_i\in K_1$ and $z_{i\rq{}}\in K_2$. Then we conclude that $G$ has an induced cycle 
$v_j,z_i,x_i,x_i',z_i'$ on 5 vertices, which contradicts Lemma~\ref{lem:cliques}. 
Therefore, there is a maximal clique $K$ in $N(v_j)$ such that $P\subseteq K\cup\{v_j\}$, and there is $z_i\in Z$  such that $z_i\in K$.
Since $z_i \notin R_j$, $v_j$ is adjacent to a vertex $x \in N(v)\setminus N[u]$ such that $x$ is not adjacent to a vertex of $D^*\setminus\{v, v_j\} $. 
Again, as $x \in N(v) \cap N(v_j)$, vertex $x \not\in P_D^*(v)$, and by Lemma~\ref{sec:linebip} we have that $N(x) \cap N(x_i) = \emptyset$ for $1 \le i \le k$.
Thus, we can conclude that $x$ is a private for $v_j$ with respect to $D$.
The obtained contradiction shows that $D$ is minimal.
\end{proof}

\begin{claim}\label{cl:B}
Each $D\in\mathcal{D}$ is a child of $D^*$ with respect to flipping $u$ and $v$.
\end{claim}

\begin{proof}[Proof of Claim~\ref{cl:B}]
To obtain a contradiction, assume that there is $D\in\mathcal{D}$ that is not a child of $D^*$. 
Let $Z$ be the set associated with $D$ when $D$ was generated.
Denote by $D^{**}$ the parent of $D$ with respect to flipping $u$ and $v$. Hence, $D=(D^{**}\setminus\{v\})\cup\{u\}\cup Z_{uv}$ and $Z\neq Z_{uv}$.
Clearly, $D^* \setminus D^{**} \subseteq Z_{uv}$, 
$D^*\setminus D^{**} \neq\emptyset$,  
$D^{**}\setminus D^{*}\subseteq Z$, and
$D^{**}\setminus D^{*} \neq\emptyset$. 
Suppose that $W=Z\cap Z_{uv}$ and $D\rq{}=((D\setminus\{u\})\cup\{v\})\setminus W$.
Observe that $D\rq{}$ is a dominating set of $G$, however it is not minimal.
Let $v_i$ be the vertex with the smallest index $i$ in   
$(D^{**}\setminus D^{*})\cup(D^{*}\setminus D^{**})$. Since $Z_{uv}$ is chosen by greedy removal of vertices, $v_i\in Z_{uv}$.
Notice also that $v_i$ has no privates with respect to $D\rq{}$ but $v_i$ has privates with respect to $D^*=D\rq{}\setminus (Z\setminus Z_{uv})$, and
$v_i$ is in $N(x)$ for some $x\in N(v)\setminus N[u]$, where $x\notin 
P_{D^*}(v)\setminus N[u]=\{x_1,\ldots,x_k\}$.
We consider two cases.

\medskip
\noindent
{\it Case a}).  Assume that $v_i$ is not adjacent to a vertex of $Z\setminus W$. 
Then because $v_i$ has no privates with respect to $D\rq{}$, $v_i$ is adjacent to another vertex $v_{i\rq{}}\in D^*\setminus\{v\}$.
Let $y\in P_{D^*}[v_i]$ and assume that $v_{j_1},\ldots,v_{j_s}\in Z\setminus Z_{uv}$ are the vertices adjacent to $y$.
Clearly,  these are the only candidates to dominate $y$.
By the construction of $Z$, each vertex $v_{j_\ell}$ is a unique element of $D$ in the neighborhood of some vertex $x_{p_\ell}\in P_{D^*}(v)\setminus N[u]$.
Vertex $y\neq v_i$ and $y\notin \{v_{j_1},\ldots,v_{j_s}\}$
as there are no edges between $v_i$ and vertices of $Z \setminus W$.
Recall that $v_{j_1},\ldots,v_{j_s}$ are dominated by $D^*$, and each vertex $v_{j_\ell}$ is a single element of $D \cap N(x_{p_\ell})$.
Since $G$ is the line graph of a bipartite graph, $y$ is adjacent either to $x$ or to some vertices $x_{p_\ell}$, 
since otherwise $xv_iyv_{j_\ell}x_{p_\ell}$ would be an induced cycle on 5 vertices, which contradicts Lemma~\ref{lem:cliques}.
Note that $y$ cannot be adjacent to both $x$ and some vertex $x_{p_\ell}$,
because otherwise 
the vertices $v_i,y,x,x_\ell$ induce a diamond, contrary to Observation~\ref{obs:diamond}.

Suppose that $y$ is adjacent to $x$, 
and hence not to any vertex of $\{x_{p_1},\ldots,x_{p_s}\}$. 
Then $\{y\}\cup\{v_{j_1},\ldots,v_{j_s}\}$ is a clique, 
because if some $v_{j_\ell},v_{j_{\ell'}}$ are not adjacent, then $N(y)$ would contain three pairwise non-adjacent vertices $x,v_{j_\ell},v_{j_{\ell'}}$ for $s> 1$, 
inducing a claw. 
By the construction of $Z$, each $v_{j_\ell}$ is dominated by a vertex of $D^*$. Then any vertex that dominates $v_{j_\ell}$, is in the maximal clique that contains $\{y\}\cup\{v_{j_1},\ldots,v_{j_s}\}$, and $y$ remains dominated after the removal of $v_{j_1},\ldots,v_{j_s}$, giving a contradiction.
Hence, $y$ is
not adjacent to $x$ but
adjacent to some vertices of $\{x_{p_1},\ldots,x_{p_s}\}$. 
Assume that $yx_{p_1}\in E(G)$. The neighborhood of $y$ is a union of at most two disjoint cliques $K_1$ and $K_2$ such that there are no edges between vertices of $K_1$ and vertices of $K_2$, by Lemma~\ref{lem:cliques}. We conclude that $s=1$.
Since $G$ is a line graph, either $v_{i\rq{}}x\in E(G)$ or $v_{i\rq{}}y\in E(G)$, as otherwise $y,v_{i\rq{}},x,v_i$ induces a claw.
If $v_{i\rq{}}x\in E(G)$, then $N[v_i]\subseteq N[ (D^*\setminus\{v,v_i\})\cup\{u\}]\cup U$, and we would set $\mathcal{D}=\emptyset$ by Case 1 i) of the algorithm.
If $v_{i\rq{}}y\in E(G)$, then $y$ cannot become a private neighbor of $v_i$ with respect to $D^*$ which is obtained from $D\rq{}$ by the removal of $Z$, again giving a contradiction.

\medskip
\noindent
{\it Case b}).  Assume that $v_i$ is adjacent to at least one vertex of $Z\setminus W$. Denote by $v_{j_1},\ldots,v_{j_s}$ these vertices. 
Observe that $i<j_1,\ldots,j_s$, because otherwise  greedy removal would remove these vertices first. 
It follows that $v_{j_1},\ldots,v_{j_s}\notin S_i$, and hence for each vertex $v_{j_t}$, at least one of the conditions in the definition of $S_i$ does not apply.  
Recall that $v_i\in D^*\setminus\{v\}$ and $v_i$ is adjacent to $x\in N(v)\setminus N[u]$. If 
(1) $x$ has no other neighbors in $D^*\setminus\{v\}$, and 
(2) $(P_{D^*}[v_i]\cap N(x))\setminus U\subseteq \{v_i\}$, then $v_{j_1},\ldots,v_{j_s}$ would be in $S_i$.
Hence, at least one of the conditions (1) and (2) is not fulfilled.

Suppose that condition (1) does not hold and there is a vertex $v_{i\rq{}}\in D^*\setminus\{v\}$ such that 
$v_{i\rq{}}\neq v_i$ and $v_{i\rq{}}\in N(x)\setminus N[v]$ as $v$ have no neighbors in $D^*$.
Since $G$ is a line graph, $v_iv_{i\rq{}}\in E(G)$. 
The set $N(v_i)$ is a union of at most two cliques and $v_{i'},v_{j_1}$ are in distinct cliques contained in $N(v_i)$.
Hence,  either $N[v_i]\subseteq N[(D^*\setminus\{v,v_i\})\cup\{u\}]\cup U$, and we would set $\mathcal{D}=\emptyset$ (see Case 1 i), or $v_{j_1}\in R_i$.
To see this, observe that if 
$N[v_i]\setminus( N[(D^*\setminus\{v,v_i\})\cup\{u\}]\cup U)\neq \emptyset$, then
$P_{D^*}[v_i]\setminus U\neq\{v_i\}$ as $v_{i\rq{}}v_i\in E(G)$.
Also $x$ is the unique vertex of $N(v_i)\cap (N(v)\setminus N[u])$, because $G$ is diamond-free by Observation~\ref{obs:diamond}.
Then $P_{D^*}(v_i)\setminus U\subseteq K$, where $K$ is the maximal clique in $N(v_i)\setminus 
N(x)$, 
and $v_{j_1}\in K$.
But by the definition, $R_i=K$, which gives a contradiction.

Therefore, we can assume now that condition (1) holds but not condition (2), and $(P_{D^*}[v_i]\cap N(x))\setminus U$ is not a subset of $\{v_i\}$. 
Then 
there is a vertex $y\in P_{D^*}[v_i]$ such that $y\neq v_i$ and $yx\in E(G)$.
Let $v_{p_1},\ldots,v_{p_l}$  be the vertices of $Z\setminus Z_{uv}$ adjacent to $y$.
By the construction of $Z$, each vertex $v_{p_r}$ is a unique element of $D$ in the neighborhood of some vertex $x_{q_r}\in P_{D^*}(v)\setminus N[u]$.
Recall that $v_{p_1},\ldots,v_{p_l}$ are dominated by $D^*$ and  each vertex $v_{p_r}$ is a single element of $D$ in $N(x_{q_r})$.
Since $G$ is the line graph of a bipartite graph, by 
Observation~\ref{obs:diamond},  $y$ is not adjacent to  $x_{p_1},\ldots,x_{p_l}$.
By the construction of $Z$, each $v_{p_r}$ is dominated by a vertex of $D^*$. Then any vertex that dominates $v_{p_r}$ is in the maximal clique that contains $\{y\}\cup\{v_{p_1},\ldots,v_{p_l}\}$ and $y$ remains dominated after the removal of $v_{p_1},\ldots,v_{p_l}$, resulting in a contradiction.
\end{proof}

\begin{claim}\label{cl:C}
If $D$ is a child of $D^*$ with respect to flipping $u$ and $v$, then  $D\in\mathcal{D}$.
\end{claim}

\begin{proof}[Proof of Claim~\ref{cl:C}]
Let $D$ be a child of $D^*$ with respect to flipping $u$ and $v$. Then $D=(D^*\setminus\{v\})\cup\{u\}\cup Z_{uv}$.
By Lemmas~\ref{lem:child-obs}, \ref{lem:parent-line} and \ref{lem:cliques}, $Z_{uv}=\{z_1,\ldots,z_k\}$, where
$z_i$ is in the clique $N(x_i)\setminus N[v]$
for $\{x_1,\ldots,x_k\}=P_{D^*}(v)\setminus N[u]$, and the cliques $N(x_i)\setminus N[v]$ are disjoint.
Also by Lemma~\ref{lem:child-obs}, each $z_i$ is adjacent to a vertex of $D^*\setminus\{v\}$.
We show that each $z_i\in Z_i$, and thus $Z_i \neq \emptyset$.

We prove that $z_i\notin R_j$ for any $j\in\{1,\dots,n\}$. To obtain a contradiction, assume that $z_i\in R_j$. Particularly, it means that $R_j\neq\emptyset$.
Recall that in this case $v_j\in D^*\setminus\{v\}$. Furthermore, if $P_{D^*}[v_j]\setminus U=\{v_j\}$ and $N(v_j)\cap (N(v)\setminus N[u])=\emptyset$, then $R_j=N(v_j)$. Moreover, if $P_{D^*}[v_j]\setminus U\neq\{v_j\}$ and (i) for any vertex $x\in N(v_j)\cap (N(v)\setminus N[u])$, $x$ is dominated by at least two vertices of $D^*\setminus\{v\}$, and (ii) $P_{D^*}(v_j)\setminus U\subseteq K$, where $K$ is a maximal clique in $N(v_i)$, then $R_i=K$.
We consider two cases.

Assume first that $P_{D^*}[v_j]\setminus U=\{v_j\}$ and
$N(v_j)\cap (N(v)\setminus N[u])=\emptyset$. Then all neighbors of $v_j$ that are not in $U$ are dominated by $D^*\setminus\{v,v_j\}$.
Therefore, all the neighbors of $v_j$ are dominated by $D\setminus\{v_j\}$.
As $z_i \in R_j$ and $R_j = N(v_j)$, in this case vertex $v_j$ has no private neighbors, contradicting the minimality of $D$.

Now let $P_{D^*}(v_j)\setminus U \neq \emptyset$. If $N(v_j)\cap (N(v)\setminus N[u])=\emptyset$
and $P_{D^*}(v_j)\setminus U\subseteq K$, where $K$ is a maximal clique in $N(v_j)$.
Then $z_i\in K$ as $R_j = K$, and we can again conclude that $D$ is not a minimal dominating set.
Suppose that $N(v_i)\cap (N(v)\setminus N[u])\neq\emptyset$ and each $x\in N(v_i)\cap (N(v)\setminus N[u])$ is dominated by
at least two vertices of $D^*\setminus\{v\}$.
If $P_{D^*}(v_j)\setminus U\subseteq K$, where $K$ is a maximal clique in $N(v_j)$, then $z_i\in K$.
As $P_D[v_j] \neq \emptyset$ and there is a vertex $z_i \in K$, we can conclude that $P_D[v_j] \subseteq N(v_j)\cap (N(v)\setminus N[u])$,
but all these vertices are dominated by $D^*\setminus\{v,v_j\}$, which contradicts the minimality of $D$.

The next step is to show that $z_i\notin S_j$ for $j\in\{1,\ldots,n\}$.
To obtain a contradiction, let $z_i\in S_j$ for some $j\in\{1,\ldots,n\}$.
Clearly, $S_j\neq \emptyset$ in this case, and recall that
$v_j$ is vertex of $D^*\setminus\{v\}$ such that (a) $v_j$  adjacent to a vertex $x\in N(v)\setminus N[u]$ that has no other neighbors in $D^*\setminus\{v\}$, and
(b) $(P_{D^*}[v_j]\cap N(x))\setminus U\subseteq \{v_j\}$.
Then $z_i=v_s$ for $s>j$ and $v_sv_j\in E(G)$.
It follows that $v_j$ is in the clique $K=N(x)\setminus N[v]$, and $K\cap D^*=\{v_j\}$,
as $N(x) \cap D^* = \{v,v_j\}$.
Consider $D'=D^*\cup\{z_i\}$. Vertex $v_j$ has a private with respect to this set, as otherwise $v_j\in D^*\setminus\{v\}$ would be removed by greedy removal of vertices
when reaching parent $D^*$ from $D$.
Notice that $v_j$ is not a private for itself because it is dominated by $z_i$.
Assume that there is a private vertex $y$ for $v_j$ such that $y\in K$.
By condition (b), all privates in $K$ should be in $U$. Hence, $y$ is either
dominated by some vertex $z_t$ or $y$ is adjacent to $u$.
Suppose first that there is a vertex $z_t$ adjacent to $y$ and notice that $z_t$ dominates a vertex $x_t$ where $x \neq x_t$.
Vertices $x$ and $x_t$ are adjacent as $N[v] \sm N(u)$ is a clique by Lemma~\ref{lem:cliques}.
By Lemma \ref{lem:child-obs}, $z_t$ is not adjacent to $v$, and $z_t$ is not adjacent to $x$, as otherwise the vertices 
$v,x_t,z_t,x$  would induce a diamond and violate Observation~\ref{obs:diamond}.
Furthermore, $y$ is not adjacent to $x_t$, as otherwise the vertices 
$x,y,z_t,x_t$ would induce a diamond.
Vertex $z_t$ is dominated by some vertex $y' \in D^*$.
Observe also that $y'$ is not adjacent to $x_t$, because $x_t$ is a private for $v$ with respect to $D^*$.
Since $G$ is a line graph and does not contain a claw as an induced subgraph, and $yx_t,x_ty' \not\in E(G)$ which implies that
$y' = v_j$ as $y$ is private for $v_j$ with respect to $D^*$,
we conclude that $yy'\in E(G)$, which contradicts our assumption that $y$ is private.
Now the vertices $x,v_j,z_t,y$ induce  a diamond which violates Observation~\ref{obs:diamond}.
Suppose now that $y$ is not dominated by $Z_{uv}$. Then $yu\in E(G)$. Since $D$ is a child of $D^*$, $u$ is adjacent to a vertex $w\in D$. Assume that $w\in D^*\setminus\{v\}$.
Notice that $y$ is not adjacent to $v$ as it is private for $v_j$.
Observe also that $w$ is not adjacent to $v$, as $v$ is private for $u$ with respect to $D$.
Since $G$ contains no induced claw, we conclude that $yw\in E(G)$.
As $y$ is private for $v_j$, this means that $v_j = w$.
Now we have a diamond induced by the vertices $u,v_j,x,y$, which violates Observation~\ref{obs:diamond}.
Hence, $w\notin D^*\setminus\{v\}$, which means that $w\in Z_{uv}$.
By the previous case $wy\notin E(G)$.
Now we have that $u$ is adjacent to $v,w,y$, but the vertices $v,w,y$ are pairwise non-adjacent, so $G$ contains a claw, which is a contradiction.
Then for any private vertex $y$ for $v_j$, $y \notin K$.
It follows that $yz_i\in E(G)$ as $G$ is a line graph.
Hence, $y$ is not a private for $v_j$ with respect to $D'$.
The obtained contradiction proves that $z_i \notin S_j$ for $j\in\{1,\ldots,n\}$.

We have thus shown that $z_i\in Z_i$ and $Z_i \neq \emptyset$ for $i\in\{1,\ldots,k\}$. Now we show that none of the three conditions of Case 1 applies, and thus we can conclude that 
$\mathcal{D} \neq \emptyset$.
(i) Since $(D^*\setminus\{v\})\cup\{u\}$ is a subset of minimal dominating set $D$, then for
any $x\in D^*\setminus\{v\}$, $N[x]$ has a vertex that is not in $U$ and not dominated by other vertices
of $D^*\setminus\{v\}$.
ii) If $k\geq 1$, then by the argument above $Z_i \neq \emptyset$ as $z_i \in Z_i$ for $i\in\{1,\ldots,k\}$.
iii) Because $D$ is a child of $D^*$, $u$ is adjacent to some vertex of $D^*\setminus\{v\}$ or
$u$ is adjacent to some vertex of $Z_{uv} \subseteq \cup_{j=1}^kZ_j$.

It remains to observe that $Z=Z_{uv}$ should be considered for the addition to $(D^*\setminus\{v\})\cap\{u\}$.
If $u$ is adjacent to some vertex of $D^*\setminus\{v\}$,
$D$ is included in $\mathcal{D}$ in Case 2, and if
$u$ is not adjacent to the vertices of $D^*\setminus\{v\}$, but $u$ is adjacent to some vertex $w\in Z_{uv}\subseteq \cup_{j=1}^kZ_j$, then
$D$ is included in $\mathcal{D}$ when we consider Case 3.
\end{proof}

To conclude the proof of Lemma \ref{lem:bipartite}, it is sufficient to observe that Claims~\ref{cl:A}, \ref{cl:B}, and \ref{cl:C} immediately imply that $\mathcal{D}$ is the set of all children of $D^*$ with respect to flipping $u$ and $v$.
\end{proof}

\section{Enumeration of minimal dominating sets of graphs of large girth}\label{sec:girth}
On line graphs we were able to observe properties of the parent relation in addition to uniqueness, which made it possible to apply the flipping method and design efficient algorithms for enumerating the minimal dominating sets. In this section we show that the flipping method can also be applied to graphs of girth at least 7. To do this, we observe other desirable properties of the parent relation on this graph class.

Let $D$ be a minimal dominating set of a graph $G$ such that $G[D]$ has at least one edge.
Let also $u\in D$ be a vertex dominated by another vertex $w\in D$ and assume that $v\in P_D(u)$.
Recall that its parent $D^*$ is defined by choosing a maximal independent set  $X_{uv}\subseteq P_D(u)\setminus N[v]$ in $G[P_D(u)\setminus N[v]]$, considering the set 
$D'=(D\setminus\{u\})\cup X_{uv}\cup \{v\}$, and then letting $D^*=D'\setminus Z_{uv}$ where $Z_{uv}\subseteq D\cap D'$.
We can easily observe the following.

\begin{lemma}\label{lem:parent-girth}
If $G$ is a graph of girth at least 7, then $X_{uv}=P_D(u)\setminus\{v\}$, and
each vertex of $Z_{uv}$ dominates at most one vertex of $\cup_{x\in X_{uv}\cup\{v\}}P_{D^*}(x)\setminus N[u]$.
\end{lemma}

Let $D^*$ be a minimal dominating set of a graph $G$ of girth at least 7
with $n$ vertices and $m$ edges, and let $v$ be an isolated vertex of $G[D^*]$. Suppose that $u$ is a neighbor of $v$.
Let $\{y_1,\ldots,y_k\}=P_{D^*}(v)\setminus N[u]$. For each $i\in\{1,\ldots,k\}$, denote by $Z_i=N(y_i)\setminus \{v\}$.
We generate a set $\mathcal{D}$  of minimal dominating sets as follows.

\medskip
\noindent
{\bf Case 1.} If $k\geq 1$ and there is an index $i\in\{1,\ldots,k\}$ such that $Z_i=\emptyset$, then $\mathcal{D}=\emptyset$.

\medskip
\noindent
{\bf Case 2.} If $k\geq 1$ and $Z_i\neq \emptyset$ for all $i\in\{1,\ldots,k\}$, then we
successively  consider all sets $Z=\{z_1,\ldots,z_k\}$ where $z_i\in Z_i$. If $k=0$, then $Z=\emptyset$. Observe that $D'=(D^*\setminus\{v\})\cup\{u\}\cup Z$ is a dominating set.
Let $W$ be the set of isolated vertices of $G[D^* \setminus\{v\}]$ belonging to $N(u)$.
We construct a partition of $W$ into three sets $X_0,X_1,X_2$ (which can be empty) as follows.
A vertex $x\in W$ is included in $X_0$ if $P_{D^*}(x)=\emptyset$, $x$ is included in
$X_1$ if $P_{D^*}(x)$ contains a vertex of degree one in $G$, and otherwise $x$ is included in $X_2$.

\medskip
\noindent
{\bf Case 2.1.} If $X_2=\emptyset$, then let $D''=D'\setminus X_0$. Observe that $D''$ is a dominating set, since 
vertices of $X_0$ are dominated by $u$, 
vertices of $X_0$ have no private neighbors outside of $D''$,
and a vertex of $X_0$ does not dominate any neighbor of another vertex of $X_0$ since $g(G)\geq 7$. 
Observe also that $X_0\cap Z=\emptyset$,
as otherwise this would make a cycle of length $5$.
We construct a minimal dominating set $D$ from $D''$ by greedy removal of vertices and add it to $\mathcal{D}$.

\medskip
\noindent
{\bf Case 2.2.} If $X_2\neq\emptyset$, then we consider all subsets $X\subseteq X_2$. Let $X=\{x_1,\ldots,x_p\}$.
Recall that for each $x_j\in X$, $P_{D'}(x_j)\neq\emptyset$. Let $P_{D'}(x_j)=\{x_j^1,\ldots,x_j^{s_j}\}$. Denote by $R_{j,t}$ the set $N(x_j^t)\setminus \{x_j\}$ for $j\in \{1,\dots,p\}$
and $t\in\{1,\ldots,s_j\}$. Note that $N(x_j^t)\setminus \{x_j\}\neq\emptyset$ because 
 $x_j^t$ has degree at least 2.

By Lemma~\ref{lem:parent-girth},
for any $i\in\{1,\ldots,k\}$, $j\in \{1,\dots,p\}$, and $t\in\{1,\ldots,s_j\}$, we have that $Z_i\cap R_{j,t}=\emptyset$, 
since otherwise we get a cycle of length $6$.
Furthermore, for any  $j,j'\in \{1,\dots,p\}$ and $t,t'\in\{1,\ldots,s_j\}$
such that $(j,t)\neq(j',t')$, we have that $R_{j,t}\cap R_{j',t'}=\emptyset$,
since otherwise we get a cycle of length $5$.
Now we consecutively consider all sets $R=\{w_{j,t} \mid 1\leq j\leq p,1\leq t\leq s_j\}$, where each vertex $w_{j,t}$ is chosen from $R_{j,t}$.
For each choice of $X$ and $R$, we construct the set $D''=(D'\setminus (X_0\cup X))\cup R$. 
The set $D''$ is a dominating set, because $u$ dominates $X_0\cup X$,
no vertex of $X_0$ has a private neighbor with respect to $D'$, and each private vertex for any vertex of $X$ with respect to $D'$ 
is dominated by a vertex of $R$.
We construct a minimal dominating set $D$ from $D''$ by greedy removal of vertices and add it to $\mathcal{D}$.

\begin{lemma}\label{lem:gen_from_par-girth}
The set $\mathcal{D}$ is a set of minimal dominating sets such that $\mathcal{D}$ contains all
children of $D^*$ with respect to flipping $u$ and $v$, and elements of $\mathcal{D}$ are generated with polynomial delay 
$O(n+m)$.
\end{lemma}

\begin{proof}
First, observe that if $\mathcal{D}\neq\emptyset$, then each $D$ is a minimal dominating set. 
Moreover, each  set $D$ constructed in Case 2 is a minimal dominating set that contains $\{u\}\cup Z\cup R$. This is because each vertex in $\{u\}\cup Z\cup R$ has a private with respect to $D''$, due to the fact that $v$ is a private for $u$, each $y_i$ is a private for $z_i$, and each $x_{j,t_j}$ is a private for $w_{j,t}$. We claim that all sets in $\mathcal{D}$ are distinct. 
For Case 2.1, the claim is straightforward, because the sets $D$ are constructed for distinct sets $Z$. 
For Case 2.2, we can observe that the vertices of $X_2\setminus X$ cannot be deleted by greedy removal, because they have private neighbors with respect to $D^*$, and these private neighbors are not dominated by $Z$, $X$, or $R$, as $g(G)\geq 7$. Hence, the sets constructed for distinct $X$ are distinct. Therefore, the sets in $\mathcal{D}$ are distinct in this case as well.  
 
Now we prove that $\mathcal{D}$ contains all children of $D^*$ with respect to flipping $u$ and $v$. Let $D=(D^*\setminus(\{v\}\cup X_{uv}))\cap\{u\}\cap Z_{uv}$.

We claim that $X_0\subseteq X_{uv}\subseteq X_0\cup X_2$.
Recall that for each vertex $x$ of $W$, $x \in X_0$ if $P_{D^*}(x)=\emptyset$,
$x\in X_1$ if $P_{D^*}(x)$ contains a vertex of degree one in $G$, and  $x\in X_2$ otherwise,
where $W$ is the set of isolated vertices of $G[D^*\setminus\{v\}]$ in $N(u)$.
We have $X_{uv}\subseteq W$, since the vertices of $X_{uv}$ are privates for $u$ in $D$. 
Suppose that $x\in X_1$. Then $x$ is adjacent to a vertex $y$ of degree one such that $y\in P_{D^*}(x)$. Vertex $y$ can be dominated either by itself or by its unique neighbor $x$. 
If 
$x\in X_{uv}$,
then $y$ has to be dominated by itself in $D$, but then $x$ is not a private for $u$ with respect to $D$.
It follows that $X_{uv}\subseteq X_0\cup X_2$. Let now $x\in X_0$. We have that $P_{D^*}(x)=\emptyset$, and since the neighbors of $x$ different from $u$ cannot be dominated by $X_{uv}\setminus\{x\}$, we conclude that $x\in X_{uv}$. 

Let $X=X_{uv}\setminus X_0$. Clearly, our algorithm considers this set.
 By Lemma~\ref{lem:child-obs},
for each $z\in Z_{uv}$,  there is a vertex  $x\in N[X_{uv}\cup\{v\}]\setminus N[u]$ adjacent to $z$
such that $x\notin N[D^*\setminus(X_{uv}\cup\{v\})]$, and  for any  $x\in N[X_{uv}\cup\{v\}]\setminus N[u]$ 
such that $x\notin N[D^*\setminus(X_{uv}\cup\{v\})]$, there is an adjacent $z\in Z_{uv}$. 
To see this, it is sufficient to observe that all neighbors of
the vertices of $X_0$ except $u$ are dominated by $D^*\setminus(X_{uv}\cup\{v\})$, because $g(G)\geq 7$.

By Lemma~\ref{lem:parent-girth}, for any $x\in N[X\cup\{v\}]\setminus N[u]$ 
such that $x\notin N[D^*\setminus(X\cup\{v\})]$, $Z_{uv}$ contains the unique vertex adjacent to $x$.
Since we consider all possible ways to dominate such vertices in our algorithm,
there are sets $Z$ and $R$ such that $Z_{uv}=Z$ in Case 2.1 or $Z_{uv}=Z\cup R$ in Case 2.2.
Consequently, $D$ is in $\mathcal{D}$.

To complete the proof, we consider the running time.
Initially all sets $Z_i$ can be generated in $O(n+m)$ time. Likewise, $W, X_0, X_1, X_2$ can be generated within the same time bound. Every set $X$ can be generated with $O(n)$ delay from the previous set $X$. For every set $X$, we need to generate the sets $R$. Before we can start generating the sets $R$, we need to generate a list  of sets $R_{j,t}$. These sets $R_{j,t}$ have empty intersections with each other, hence the sum of the sizes of all sets $R_{j,t}$ is $O(n)$ for each $X$. Thus, after generating $X$, we can generate the list of all $R_{j,t}$ in $O(n+m)$ time. Now, as long as $X$ is fixed, we can generate every $R$ in time $O(n)$ using the list of $R_{i,j}$. Note that every $R$ gives us a new dominating set $D$. This means that the delay between each generated dominating set $D$ is $O(n+m)$. 
\end{proof}

Combining Lemmas~\ref{lem:main} and \ref{lem:gen_from_par-girth}, we obtain the following theorem.

\begin{theorem}\label{thm:girth}
All minimal dominating sets of a graph of girth at least 7 with $n$ vertices and $m$ edges can be enumerated
with incremental-polynomial delay 
$O(n^2m|\mathcal{L}|^2)$, 
and in total time 
$O(n^2m|\mathcal{L}^*|^2)$,
where $\mathcal{L}$ is the set of already generated minimal dominating sets
and $\mathcal{L}^*$ is the set of all minimal dominating sets.
\end{theorem}

\end{document}